\documentclass[11pt]{article}
\usepackage{graphicx}
\usepackage{subcaption}
\usepackage{todonotes}
\usepackage{mathtools}
\usepackage{soul}
\usepackage{amsthm}
\usepackage{amsmath,amssymb}
\usepackage{systeme}

\DeclareMathOperator*{\argmin}{arg\,min}
\usepackage{gensymb}
\usepackage[ruled,vlined]{algorithm2e}
\SetKwComment{Comment}{/* }{ */}
\SetKwInOut{Input}{input}\SetKwInOut{Output}{output}

\graphicspath{{./figs/}}
\usepackage[absolute]{textpos}

\usepackage{fullpage,cite}
\usepackage[top=0.8in, bottom=0.9in, left=0.9in, right=0.9in]{geometry}

\usepackage[pdftex, plainpages = false, pdfpagelabels,
                 bookmarks=false,
                 bookmarksopen = true,
                 bookmarksnumbered = true,
                 breaklinks = true,
                 linktocpage,
                 pagebackref,
                 colorlinks = true,  
                 linkcolor = blue,
                 urlcolor  = cyan,
                 citecolor = red,
                 anchorcolor = green,
                 hyperindex = true,
                 hyperfigures
                 ]{hyperref}

\def\calP{\mathcal{P}}
\def\calK{\mathcal{K}}
\def\calI{\mathcal{I}}

\def\calD{\mathcal{D}}
\def\calR{\mathcal{R}}
\def\lowenv{\text{LE}}
\def\st{$s$-$t$}

\newtheorem{theorem}{Theorem}
\newtheorem{corollary}{Corollary}
\newtheorem{lemma}{Lemma}
\newtheorem{observation}{Observation}

\title{Shortest Paths on Convex Polyhedral Surfaces\thanks{A preliminary version of this paper will appear in {\em Proceedings of the 66th IEEE Symposium on Foundations of Computer Science (FOCS 2025)}~\cite{ref:WangSh25}. This version further improves the results presented in the preliminary version.}
}
\author{Haitao Wang\thanks{Kahlert School of Computing,
University of Utah, Salt Lake City, UT 84112, USA. {\tt haitao.wang@utah.edu}}
}
\date{}

\begin{document}


\maketitle

\vspace{-0.3in}
\begin{abstract}
Let $\mathcal{P}$ be the surface of a convex polyhedron with $n$ vertices. We consider the two-point shortest path query problem for $\mathcal{P}$: Constructing a data structure so that given any two query points $s$ and $t$ on $\mathcal{P}$, a shortest path from $s$ to $t$ on $\mathcal{P}$ can be computed efficiently. To achieve $O(\log n)$ query time (for computing the shortest path length), 
the previously best result uses $O(n^{8+\epsilon})$ preprocessing time and space [Aggarwal, Aronov, O’Rourke, and Schevon, SICOMP 1997], where $\epsilon$ is an arbitrarily small positive constant. In this paper, we present a new data structure of $O(n^{6+\epsilon})$ preprocessing time and space, with $O(\log n)$ query time. For a special case where one query point is required to lie on one of the edges of $\mathcal{P}$, the previously best work uses $O(n^{6+\epsilon})$ preprocessing time and space to achieve $O(\log n)$ query time. We improve the preprocessing time and space to $O(n^{5+\epsilon})$, with $O(\log n)$ query time.
Furthermore, we present a new algorithm to compute the exact set of shortest path edge sequences of $\mathcal{P}$, which are known to be $\Theta(n^4)$ in number and have a total complexity of $\Theta(n^5)$ in the worst case. The previously best algorithm for the problem takes roughly $O(n^6\log n\log^*n)$ time, while our new algorithm runs in $O(n^{5+\epsilon})$ time. 
\end{abstract}



\section{Introduction}
\label{sec:intro}

Let $\calP$ denote the surface of a convex polyhedron with $n$ vertices. We can represent $\calP$ by a planar graph with additional geometric information describing the location of vertices and edges. By Euler's formula, the
number of edges and faces in $\calP$ is $O(n)$.  Each face of $\calP$ is a convex polygon. We can assume that each face is associated with a coordinate system on its supporting plane, i.e., the plane containing the face. 
For any two points $s$ and $t$ on $\calP$, let $\pi(s,t)$ denote a shortest path from $s$ to $t$ on $\calP$, and its length, denoted by $d(s,t)$, is called the {\em geodesic distance} between $s$ and $t$. We use $|\pi(s,t)|$ to denote the number of edges in the path. 

\paragraph{\bf Two-point shortest path queries.}
We consider the {\em two-point shortest path query problem} for $\calP$: Construct a data structure so that a shortest \st\ path $\pi(s,t)$ can be computed efficiently for any two query points $s,t\in \calP$. The previously best result was given by Aggarwal, Aronov, O’Rourke, and Schevon~\cite{ref:AgarwalSt97}. To achieve $O(\log n)$ query time to compute the geodesic distance $d(s,t)$, their data structure requires $O(n^{8+\epsilon})$ space and $O(n^{8+\epsilon})$ preprocessing time. Throughout the paper, $\epsilon$ represents an arbitrarily small positive constant. Their query algorithm does not output the path $\pi(s,t)$. 
No further progress had been made on this problem in the 35 years since their result was first announced in 1990. 
In this paper, we improve preprocessing time and space to $O(n^{6+\epsilon})$ while still achieving the $O(\log n)$ query time for computing $d(s,t)$. In addition, our query algorithm can output $\pi(s,t)$ in additional $O(|\pi(s,t)|)$ time. 

\paragraph{\bf The edge-face case.} In the {\em edge-face case}, one query point is required to lie on one of the edges of $\calP$. For this case, the previously best result was also given in \cite{ref:AgarwalSt97}. To achieve $O(\log n)$ query time for computing $d(s,t)$, the data structure in \cite{ref:AgarwalSt97} uses $O(n^{6+\epsilon})$ space and preprocessing time. We present a new data structure of $O(n^{5+\epsilon})$ space and preprocessing time, with $O(\log n)$ query time. In addition, $\pi(s,t)$ can be output in additional $O(|\pi(s,t)|)$ time. 
In particular, given an edge $e$ of $\calP$, we can construct a data structure of $O(n^{4+\epsilon})$ space and preprocessing time, such that for any query $(s,t)$ with $s\in e$ and $t\in \calP$, $d(s,t)$ can be computed in $O(\log n)$ time and  $\pi(s,t)$ can be output in additional $O(|\pi(s,t)|)$ time.

Note that Cook IV and Wenk~\cite{ref:CookSh12} considered a special case in which $s$ is required to be on a given edge $e$ of $\calP$ and $t$ is required to be on another edge $\calP$. They constructed a data structure of $O(n^5)$ space in $O(n^5\log n)$ time such that $d(s,t)$ can be computed in $O(\log^2 n)$ time for any query $(s,t)$. Hence, our above result, which works for a more general type of queries, even improves upon their work in this restricted setting. 


\paragraph{\bf Shortest path edge sequences.} For any shortest path $\pi(s,t)$ on $\calP$, it has the following properties~\cite{ref:SharirOn86}: (1) $\pi(s,t)$ can cross every edge of $\calP$ at most once; (2) the intersection between $\pi(s,t)$ and any face of $\calP$ is a single line segment; (3) the interior of $\pi(s,t)$ cannot contain any vertex of $\calP$. Furthermore, there is an {\em unfolding property}: If we unfold all faces of $\calP$ intersected by $\pi(s,t)$ into one plane, then $\pi(s,t)$ becomes a line segment in the unfolded plane. The sequence of edges of $\calP$ crossed by $\pi(s,t)$ is called the {\em edge sequence} of $\pi(s,t)$. Once the edge sequence of $\pi(s,t)$ is known, $\pi(s,t)$ can be easily obtained due to the unfolding property. 

A sequence of edges of $\calP$ is a {\em shortest path edge sequence} 
if it is the sequence of edges crossed by a shortest path on $\calP$. It is known that the number of shortest path edge sequences of $\calP$ is $O(n^4)$ and this bound is tight in the worst case~\cite{ref:MountTh90}. Computing all shortest path edge sequences will be useful for understanding the structure of shortest paths on $\calP$ and also useful for designing algorithms to compute obstacle-avoiding paths among polyhedral obstacles in the 3D space~\cite{ref:SharirOn87,ref:SchevonTh88,ref:SchevonAl89}. 

Schevon and O'Rourke~\cite{ref:SchevonAl89} gave an algorithm that can compute all shortest path edge sequences in $O(n^9\log n)$ time. Based on a so-called {\em star unfolding}, Agarwal, Aronov, O'Rourke, and Schevon~\cite{ref:AgarwalSt97} gave an improved algorithm of roughly $O(n^6\log n\log^* n)$ time, and later Cook IV and Wenk~\cite{ref:CookSh12} devised an alternative algorithm of the same runtime using kinetic Voronoi diagrams~\cite{ref:GuibasVo91}. 
In this paper, we present a new algorithm of $O(n^{5+\epsilon})$ time. This is the first progress on the runtime for the problem since the algorithm of \cite{ref:AgarwalSt97} was first announced in 1990. 

Note that since the number of shortest path edge sequences can be $\Theta(n^4)$ and each sequence can have $\Omega(n)$ edges in the worst case~\cite{ref:MountTh90}, it requires $\Omega(n^5)$ time to explicitly list all edges in all edge sequences. Therefore, our algorithm, which runs in $O(n^{5+\epsilon})$ time, is almost optimal for that purpose. 

Note that all above algorithms are for computing the exact set of shortest path edge sequences. Computing a superset of them is easier and algorithms were also known~\cite{ref:AgarwalSt97,ref:CookSh12,ref:SharirOn87}.

\subsection{Related work}
Shortest paths on polyhedral surfaces have been extensively studied in the literature~\cite{ref:AgarwalSt97,ref:SharirOn86,ref:MitchellTh87,ref:MountOn85,ref:SchreiberAn08,ref:ChenSh90,ref:AronovNo92,ref:CookSh12,ref:MountSt87,ref:ChandruSh04,ref:HWangFi89,ref:KapoorEf99}. Sharir and Schorr~\cite{ref:SharirOn86} first considered the problem of computing a single shortest path on a convex polyhedral surface $\calP$ and gave an $O(n^3\log n)$ time algorithm. Mitchell, Mount, and Papadimitriou~\cite{ref:MitchellTh87,ref:MountOn85} improved it to $O(n^2\log n)$ time. Chen and Han~\cite{ref:ChenSh90} later presented an $O(n^2)$ time solution. Schreiber and Sharir~\cite{ref:SchreiberAn08} finally derived an algorithm that runs in $O(n\log n)$ time. Given a source point $s$, the algorithm of \cite{ref:SchreiberAn08} can construct a data structure of $O(n\log n)$ space for $\calP$ in $O(n\log n)$ time so that each {\em one-point shortest path query} (i.e., given a query point $t\in \calP$, compute the geodesic distance $d(s,t)$) can be answered in $O(\log n)$ time. All these algorithms use the continuous Dijkstra approach except that Chen and Han's algorithm~\cite{ref:ChenSh90} is based on a so-called ``one angle one split'' property to bound the size of a sequence tree. 

If $\calP$ is not convex, then finding a shortest path on $\calP$ is more challenging. A notable difference is that in this case a shortest path may pass through a vertex of $\calP$ in its interior. The algorithm of \cite{ref:MitchellTh87} can compute a shortest path on $\calP$ in $O(n^2\log n)$ time, and the $O(n^2)$ time algorithm of \cite{ref:ChenSh90} also works for the non-convex case. Both algorithms can construct a data structure of $O(n^2)$ space for a source point $s$ within the same time complexities as above such that each one-point shortest path query can be answered in $O(\log n)$ time. 
Kapoor~\cite{ref:KapoorEf99} announced an algorithm of $O(n\log^2 n)$ time for the problem; but the details of the algorithm have not yet been published.
For two-point shortest path queries in the non-convex case, the only work that we are aware of was due to Chiang and Mitchell~\cite{ref:ChiangTw99}. To achieve $O(\log n)$ query time, their data structure uses $O(n^{12})$ space.

Approximation algorithms have also been developed for the shortest path problems on convex polyhedral surfaces; e.g., see \cite{ref:ChengSh14,ref:VaradarajanAp00}.

Given a set of polyhedral obstacles in 3D space, an even more challenging problem is to compute a shortest path between two points  while avoiding all obstacles. The problem is NP-hard~\cite{ref:CannyNe87} and approximation algorithms have been developed~\cite{ref:SharirOn86,ref:SharirOn87}. A key subproblem in these algorithms involves computing shortest paths on the surface of a convex polyhedron. In particular, an important subproblem in these algorithm is the edge-edge case shortest path problem discussed above. 


The obstacle-avoiding shortest path problem in the plane has also been studied extensively, e.g., \cite{ref:GhoshAn91,ref:KapoorAn97,ref:MitchellA91,ref:MitchellSh96,ref:HershbergerAn99,ref:RohnertSh86}. Given a set of polygonal obstacles of $n$ vertices in the plane, the problem is to compute a shortest path between two points that avoid all obstacles. 
With respect to a source point, Hershberger and Suri~\cite{ref:HershbergerAn99} built a data structure of $O(n)$ space that can answer each one-point shortest path query in $O(\log n)$ time. Their algorithm runs in $O(n\log n)$ time and space. The space has been improved to $O(n)$ recently~\cite{ref:WangA23}. Like the polyhedral surface case, the two-point shortest path query problem becomes significantly more complicated. Chiang and Mitchell~\cite{ref:ChiangTw99} built a data structure of $O(n^{11})$ size with  $O(\log n)$ query time. There had been no progress on the problem until very recently in SoCG 2024 de Berg, Miltzow, and Staals~\cite{ref:deBergTo24} proposed a data structure of size $O(n^{10+\epsilon})$ with $O(\log n)$ query time. 

\subsection{An overview of our approach}
\label{sec:overview}

We give a brief overview of our approach, using some concepts that were introduced previously but will also be explained in more detail later in the paper. 

\paragraph{\bf Two-point shortest path queries.}
For the general case two-point shortest path query problem, we follow the high-level idea of Aggarwal, Aronov, O’Rourke, and Schevon~\cite{ref:AgarwalSt97} (referred to as AAOS algorithm) but replace their ``brute force'' approach by a more efficient way to ``re-use'' certain information; this strategy saves a quadratic factor of the preprocessing time and space. More specifically, the ridge trees of all vertices of $\calP$ partition $\calP$ into $O(n^4)$ ridge-free regions such that all points $s$ in each ridge-free region $R$ have a topologically equivalent star unfolding $\bigstar_R$. To achieve $O(\log n)$ query time, by using the star unfolding $\bigstar_R$, the AAOS method constructs a data structure of $O(n^{4+\epsilon})$ size for each ridge-free region $R$ to handle the queries $(s,t)$ with $s\in R$. As there are $O(n^4)$ ridge-free regions, the total size of the overall data structure is $O(n^{8+\epsilon})$.

Our new idea is motivated by the following key observation. Let $R$ and $R'$ be two neighboring ridge-free regions that share a common edge. Their star unfoldings $\bigstar_R$ and $\bigstar_{R'}$ are topologically nearly the same: They only differ by at most four vertices~\cite{ref:CookSh12}. As such, if we already have a data structure for $R$ based on $\bigstar_R$, then it is a waste of resources to construct a ``brand-new'' data structure for $\bigstar_{R'}$, and instead, we could ``re-use'' certain information of the data structure for $\bigstar_R$ to construct that for $\bigstar_{R'}$ (this resembles the general principle of persistent data structures~\cite{ref:DriscollMa89}). Rather than directly follow this idea, we build a hierarchical structure so that if a data structure for a big region $R$ is already constructed, then for each subregion $\tilde{R}\subset R$, when constructing the data structure for $\tilde{R}$, we only need to consider the information that was not already used for $R$. This is our strategy to save resources. To implement this idea, we construct cuttings~\cite{ref:ChazelleCu93,ref:WangUn23} on the edges of all ridge-trees. 


Note that the preliminary version of this work~\cite{ref:WangSh25} used a single cutting and achieved $O(n^{6+\frac{2}{3}+\epsilon})$
preprocessing time and space. Here we improve the preprocessing to $O(n^{6+\epsilon})$ by employing a hierarchy of $O(1)$ cuttings. 

\paragraph{\bf The edge-face case.}
For the edge-face case, we follow a similar idea as in the above general case. 
The ridge trees of all vertices of $\calP$ partition the edges of $\calP$ into $O(n^3)$ segments, called {\em edgelets}~\cite{ref:AgarwalSt97}, such that all points in the same edgelet $\eta$ have a topologically equivalent star unfolding $\bigstar_{\eta}$. The AAOS method constructs a data structure of $O(n^{3+\epsilon})$ size for each edgelet $\eta$ to handle the queries $(s,t)$ with $s\in \eta$. As there are $O(n^3)$ edgelets, the total size of the overall data structure is $O(n^{6+\epsilon})$.

Our method is motivated by a similar observation to the above general case: If $\eta$ and $\eta'$ are two adjacent edgelets on the same edge of $\calP$, then their star unfoldings differ by at most four vertices. We again make use of cuttings as above. But since we are constructing cuttings on each edge of $\calP$, which is one dimensional, a cutting is just a subdivision of each edge of $\calP$ into segments. Therefore, the algorithm is much easier in this case. 

Note that for this case the preliminary version~\cite{ref:WangSh25} used a single cutting and achieved $O(n^{5+\frac{1}{4}+\epsilon})$
preprocessing time and space. Here we instead use a hierarchy of $O(1)$ cuttings and improve the preprocessing to $O(n^{5+\epsilon})$.


\paragraph{\bf Shortest path edge sequences.}
Instead of using the star unfolding as in \cite{ref:AgarwalSt97} or using kinetic Voronoi diagrams as in~\cite{ref:GuibasVo91}, we propose a new method to tackle the problem. Suppose a point $s$ is on an edge $e$ of $\calP$. We observe that the set $\Sigma_s$ of shortest path edge sequences of $\pi(s,t)$ for all points $t\in \calP$ is determined by the ridge tree $T_s$ of $s$. When $s$ moves on $e$, $T_s$ also changes. However, $\Sigma_s$ remains invariant as long as $T_s$ does not undergo any topological change. We devise a technique to track the ``events'' that trigger these topological changes of $T_s$. Our algorithm runs in $O(K\cdot \log n)$ time, where $K$ is the number of events. 

Establishing a good upper bound for $K$ turns out to be one of the biggest challenges of this paper. Specifically, the problem reduces to bounding the number $k_{e'}$ of events caused by $e'$, for any edge $e'\in \calP$. Traditional methods such as the source unfolding~\cite{ref:SharirOn86,ref:MountOn85} or star unfolding~\cite{ref:AgarwalSt97,ref:ChenSh90,ref:AronovNo92} do not seem to work. We instead propose a new unfolding method that is to unfold all faces intersecting all possible shortest paths $\pi(s,t)$ with $s\in e$ and $t\in e'$ to a plane containing $e'$. The new unfolding is inspired by the ``one angle one split'' property in \cite{ref:ChenSh90}. Most importantly, and analogously to the star unfolding, we prove that there is no ``short-cut'' in the unfolding, i.e., for any image $s'$ of $s$ and any point $t\in e'$, the length of the line segment $\overline{s't}$ in the unfolded plane is no shorter than the length of $\pi(s,t)$. Our new unfolding method may be of independent interest. 
Using our new unfolding, we prove that $k_{e'}=O(n^{3+\epsilon})$. This leads to $K=O(n^{4+\epsilon})$ and further yields an $O(n^{5+\epsilon})$ time algorithm to compute the exact set of shortest path edge sequences of $\calP$.

\paragraph{\bf The edge-edge case two-point shortest path queries.}
Our techniques for computing shortest path edge sequences might be useful for other problems. We demonstrate one application on the {\em edge-edge case} two-point shortest path query problem in which both query points $(s,t)$ are required on edges of $\calP$. We obtain a result matching our edge-face bound, using a simpler approach. 

For a fixed point $s$ on an edge $e$ of $\calP$, using its ridge tree $T_s$, we can easily compute $d(s,t)$ in $O(\log n)$ time for any query point $t\in \calP$~\cite{ref:SharirOn86,ref:MountOn85}. To solve the two-point query problem, we consider the topological changes of $T_s$ as $s$ moves along $e$. We show that $e$ can be partitioned into $O(n^{4+\epsilon})$ intervals such that $T_s$ is topologically equivalent for all points $s$ in the same interval. This result follows directly from our earlier proof that 
$K=O(n^{4+\epsilon})$. Using this partitioning and a persistent data structure~\cite{ref:SarnakPl86}, we construct a data structure in $O(n^{4+\epsilon})$ preprocessing time and space for the edge $e$ so that each two-point query $(s,t)$ with $s\in e$ and $t$ on any edge of $\calP$ can be answered in $O(\log n)$ time. 

\paragraph{Outline.} The remainder of the paper is organized as follows. After introducing some notation and concepts in Section~\ref{sec:pre}, we present our data structure for the two-point shortest path query problem in Section~\ref{sec:twopoint}. The edge-face case is discussed in Section~\ref{sec:edge}. Our algorithm for computing the shortest path edge sequences is given in Section~\ref{sec:seq}. We finally present a solution to the edge-edge case shortest path queries in Section~\ref{sec:edgeboth} as an application of our techniques in Section~\ref{sec:seq}.

\section{Preliminaries}
\label{sec:pre}

We follow the notation in Section~\ref{sec:intro}, e.g., $\calP$, $n$, $\pi(s,t)$, $d(s,t)$, etc. 
We assume that each face of $\calP$ is a triangle since otherwise we could triangulate every face. We refer to the vertices/edges/faces of $\calP$ as {\em polyhedron vertices/edges/faces}. 

For any two points $p$ and $q$, we use $\overline{pq}$ to denote the line segment connecting $p$ and $q$, and use $\Vert pq\Vert$ to denote the (Euclidean) length of $\overline{pq}$. 

A path $\pi$ on $\calP$ is a {\em geodesic path} if it cannot be shortened by a local change at any point in its relative interior. A geodesic path $\pi$ also has the {\em unfolding property}: It is possible to unfold all faces of $\calP$ intersected by $\pi$ in one plane $\Pi$ so that $\pi$ becomes a line segment in $\Pi$. If $\sigma$ is the edge sequence of $\pi$, we say that the above unfolding is the {\em unfolding of $\sigma$ on the plane $\Pi$}. The unfolding property implies that a geodesic path $\pi$ does not contain any polyhedron vertex in its interior~\cite{ref:AgarwalSt97}. 

\paragraph{\bf Ridge trees.}
Consider a point $s\in \calP$. A point $t\in \calP$ is a {\em ridge point} of $s$ if there exist at least two shortest paths from $s$ to $t$ on $\calP$. Ridge points of $s$ constitute $O(n^2)$ line segments on $\calP$, which together form a {\em ridge tree} $T_s$ of at most $n$ leaves~\cite{ref:SharirOn86}; see Figure~\ref{fig:ridgetree}. Each leaf of $T_s$ is a vertex of $\calP$. But a vertex $v$ of $\calP$ may not be a leaf of $T_s$ if $s$ lies on the ridge tree of $v$, in which case $v$ becomes a degree-2 vertex of $T_s$. Since $T_s$ has at most $n$ leaves, it has $O(n)$ vertices of degrees $3$ or higher. A vertex of degree $3$ or higher in $T_s$ is called a {\em high-degree vertex}. If $s$ is in a degenerate position, then $T_s$ may have a vertex of degree larger than $3$. 
Following \cite{ref:AgarwalSt97}, we define a {\em ridge} as a maximal subset of $T_s$ that does not contain a polyhedron vertex or a high-degree vertex. It is known that $T_s$ has $O(n)$ ridges and each of them is a shortest path on $\calP$~\cite{ref:AgarwalSt97}. Hence, each ridge of $T_s$ intersects a polyhedral edge at most once and $T_s$ intersects each polyhedral edge $O(n)$ times. Therefore, there are at most $O(n^2)$ intersections between $T_s$ and all edges of $\calP$; we also consider these intersections as {\em vertices} of $T_s$, which are of degree $2$. Hence, $T_s$ has $O(n)$ leaves and high-degree vertices, and $O(n^2)$ degree-2 vertices.

\begin{figure}[t]
\begin{minipage}[t]{\linewidth}
\begin{center}
\includegraphics[totalheight=1.0in]{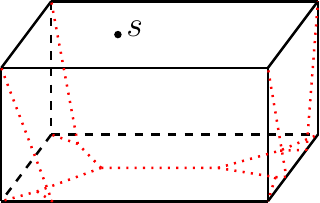}
\caption{
Illustrating the ridge tree of a point $s$ on a convex polyhedron.}
\label{fig:ridgetree}
\end{center}
\end{minipage}
\end{figure}

Another important property of shortest paths is the following: For any point $t\in \calP$, the interior of $\pi(s,t)$ does not contain any ridge point of $s$, i.e., does not intersect $T_s$~\cite{ref:SharirOn86}. 
Using $T_s$, one can obtain a so-called {\em source unfolding} by cutting $\calP$ along edges of $T_s$ and embed it in a plane~\cite{ref:ORourkeCo89,ref:SharirOn86}. 

A point of $\calP$ is called a {\em generic point} if it is not a ridge point of any polyhedron vertex. For each face $f$ of $\calP$, the maximal connected portion of $f$ consisting of generic points is called a {\em ridge-free region}; for each polyhedron edge $e$, the maximal connected portion of $f$ consisting of generic points is called an {\em edgelet}~\cite{ref:AgarwalSt97}. The number of edgelets on $\calP$ is $\Theta(n^3)$ in the worst case and they can be computed in $O(n^3\log n)$ time~\cite{ref:AgarwalSt97}; the number of ridge-free regions on $\calP$ is $\Theta(n^4)$ in the worst case and they can be computed in $O(n^4)$ time~\cite{ref:AgarwalSt97}.

A point of $\calP$ is a {\em $0$, $1$, and $2$-dimensional} point if it is a polyhedron vertex, in the interior of a polyhedron edge, and in the interior of a polyhedron face, respectively. For ease of exposition, we make a general position assumption that for an $i$-dimensional point $s$ and a $j$-dimensional point $t$, there are at most $k$ shortest \st\ paths on $\calP$, with $k=1+i+j$ for any $0\leq i, j\leq 2$. For example, by the assumption, $T_s$ does not have a vertex of degree larger than 
$4$ for any point $s$ on a polyhedron edge; 
also, no two points on  polyhedron edges have more than three shortest paths. 

\begin{figure}[t]
\begin{minipage}[t]{\linewidth}
\begin{center}
\includegraphics[totalheight=1.8in]{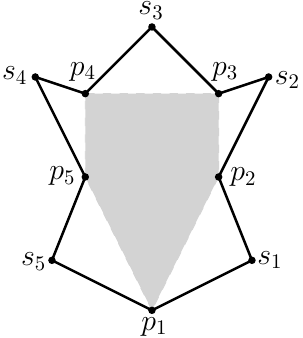}
\caption{
Illustrating $\bigstar_s$: The gray region is its kernel $\calK_s$; the points $p_i$, $1\leq i\leq 5$, are the images of the vertices of $\calP$ while the points $s_i$, $1\leq i\leq 5$, are the images of $s$.}
\label{fig:star}
\end{center}
\end{minipage}
\end{figure}

\paragraph{\bf Star unfolding.}
Let $s$ be a generic point. If we cut $\calP$ open by shortest paths from $s$ to all vertices of $\calP$, then we obtain a two-dimensional complex, which can be unfolded into a (non-self-overlapping) simple polygon in a single plane~\cite{ref:AronovNo92}; this is called the {\em star unfolding} of $\calP$~\cite{ref:AgarwalSt97,ref:AronovNo92} and we use $\bigstar_s$ to denote the simple polygon.
We discuss some properties of $\bigstar_s$ that will be used later in our algorithm. These properties have been proved previously~\cite{ref:AgarwalSt97,ref:ChenSh90,ref:AronovNo92}; refer to the previous work (e.g., \cite{ref:AgarwalSt97}) for more details. 

For each point $p\in \calP$, if $p=s$, then $p$ has $n$ images in $\bigstar_s$ (called {\em source images}). If $p$ lies in the interior of the shortest path from $s$ to a polyhedron vertex, then $p$ has two images. In all other cases (in particular, when $p$ is a polyhedron vertex), $p$ has a unique image. The simple polygon $\bigstar_s$ can be partitioned into $n$ triangles and a simple polygon, called {\em kernel} and denoted by $\calK_s$; see Figure~\ref{fig:star}. Each triangle $\triangle$ shares a single edge with $\calK_s$ and does not share any edge with other triangles. 
We refer to these triangles $\triangle$ as {\em source triangles}. 

For any point $t\in \calP$, shortest path $\pi(s,t)$ maps to a line segment on $\bigstar_s$ connecting an image of $s$ and an image of $t$~\cite{ref:AgarwalSt97}. The star unfolding $\bigstar_s$ also has the following ``no short-cut'' property: For any point $t^*\in \bigstar_s$ and any source image $s^*$, $|s^*t^*|\geq d(s,t)$, where $t$ is a preimage of $t^*$ on $\calP$~\cite{ref:AronovNo92} (refer to \cite{ref:ChandruSh04} for some stronger ``no short-cut'' properties). 


Let $p_1,p_2,\ldots,p_n$ be the images of the vertices of $\calP$ that ordered cyclically around the boundary of $\bigstar_s$. 
Then, it is possible to name the source images as $s_1,s_2,\ldots,s_n$ such that $\bigstar_s$ is cyclically bounded by the segments connecting every two
adjacent points of the following cyclic list: $s_1,p_1,s_2,p_2,\ldots,s_n,p_n$ (see Figure~\ref{fig:star}).
The kernel $\calK_s$ is cyclically bounded by the segments connecting every two
adjacent points of $p_1,p_2,\ldots,p_n$.
Although $\bigstar_s$ is bounded by $2n$ segments in the unfolded plane, its combinatorial complexity
is $\Theta(n^2)$. Specifically, a vertex of $\bigstar_s$ is $s_i$, $p_i$, or the image of an
intersection between an edge of $\calP$ and the shortest path from $s$ to a polyhedron vertex. 
An edge of $\bigstar_s$ is a maximal portion of the image of a polyhedron edge or the shortest path from $s$ to a polyhedron vertex
delimited by vertices of $\bigstar_s$. Each face of $\bigstar_s$ is called a {\em plate}. It is known that $\bigstar_s$ has
$O(n^2)$ vertices, edges, and plates; this bound is tight in the worst case. 

We define the {\em combinatorial structure} of $\bigstar_s$ as the graph whose nodes and arcs are the vertices and edges of
$\bigstar_s$, respectively. It is known that that $\bigstar_s$ is topologically equivalent (i.e., has the same
combinatorial structure) for all points $s$ in the same ridge-free region~\cite{ref:AgarwalSt97}, and furthermore, the kernel $\calK_s$ is fixed for all points $s$ in the same ridge-free region. Hence, we can use $\bigstar_R$ to denote the combinatorial structure of the star unfolding for all points in a ridge-free region $R$ and use $\calK_R$ to denote its kernel. 

\paragraph{\bf Cuttings.}
Let $S$ be a set of $n$ line segments in the plane. Let $m$ be the number of intersections of these segments. 
For a parameter $r$ with $1 \leq r \leq n$, a {\em $(1/r)$-cutting} for $S$ is a collection $\Xi$ of cells (each of which is a possibly-unbounded trapezoid) with disjoint interiors whose union covers the entire plane such that the interior of every cell $\sigma\in\Xi$ is intersected by at most $n/r$ segments of $S$. 
The {\em size} of $\Xi$ is defined to be the number of cells of $\Xi$.
A $(1/r)$-cutting of size $O(r^2)$ exists for $S$ and can be computed in $O(nr)$ time~\cite{ref:ChazelleCu93}.
The algorithm also produces the list of segments intersecting each cell of the cutting.

Furthermore, cuttings whose sizes are sensitive to $m$ can also be constructed. Specifically, a $(1/r)$-cutting of size $O(r^{1+\epsilon}+m\cdot r^2/n^2)$ exists for $S$\cite{ref:AgarwalPs05,ref:ChazelleCu93,ref:WangUn23} and can be computed in $O(nr^{\epsilon}+mr/n)$ time~\cite{ref:ChazelleCu93,ref:WangUn23}. The algorithm also produces the list of segments intersecting each cell of the cutting. Note that $nr^{\epsilon}+mr/n=O(nr)$ and we will use the time bound $O(nr)$ in our algorithm, which is sufficient for our purpose.

\section{Two-point shortest path queries -- the general case}
\label{sec:twopoint}

In this section, we present our data structure for the general case of the two-point shortest path queries. Given a pair of query points $(s,t)$ on $\calP$, our goal is to compute a shortest path $\pi(s,t)$. Our discussion will first focus on computing the geodesic distance $d(s,t)$ and then explain how to report $\pi(s,t)$. 
At a high level, our approach builds on the AAOS algorithm~\cite{ref:AgarwalSt97}. We begin with a brief review of the AAOS algorithm and then discuss our improvements in detail.

\subsection{A review of the AAOS algorithm}
\label{sec:review}

In the preprocessing, we compute ridge trees for all vertices of $\calP$. Then,
we obtain all ridge-free regions. All these can be done in $O(n^4)$ time. 

Consider two points $s,t\in \calP$. We first assume that $s$ is a generic point contained in a ridge-free region $R$. 
Let $t^*$ denote the image of $t$ in the star unfolding $\bigstar_s$. 
Depending on whether $t^*$ is in the kernel $\calK_R$, there are two cases. 

\paragraph{The outside-kernel case $\boldsymbol{t^*\not\in \calK_R}$.}
If $t^*\not\in \calK_R$, then $t^*$ lies in a source triangle $\triangle$ of
$\bigstar_s$. This case can be easily handled due to the following observation. 
Let $\Phi_R$ denote the preimage of the boundary of $\calK_R$ on $\calP$. Let $f$ be a face of $\calP$ and $C$ be a connected
component of $f\setminus \Phi_R$ whose image in $\bigstar_s$ is not contained in $\calK_R$. 
The observation is that the edge sequence of $\pi(s,t)$ for all $s\in
R$ and $t\in C$ is the same~\cite{ref:AgarwalSt97}. To answer queries in
this case, we perform the following preprocessing work. 

For each ridge-free region $R$, we choose an arbitrary point $p_R$ in $R$.  
We compute its kernel $\calK_R$ and obtain the preimage $\Phi_R$ of the boundary of $\calK_R$ on $\calP$. 
For each polyhedron face $f$, we compute
the connected components of $f\setminus \Phi_R$. For each component $C$, we label it 
whether its image is in $\calK_R$, and if not, we pick an
arbitrary point $q_C$ as its ``representative'' point. 
According to the above observation, the edge sequence of $\pi(s,t)$ is the same as that of $\pi(p_R,q_C)$
for any point $s\in R$ and any point $t\in C$. The total number of such
connected components in all the faces of $\calP$ is $O(n^2)$. As such, we obtain
$O(n^2)$ representative points on $\calP$. For each connected component $C$ in a polyhedron face $f$, we
compute the coordinate transformation, corresponding to the unfolding of the edge sequence of
$\pi(p_R,q_C)$, which maps the $f$-based coordinates of points in $f$ to the
$f_R$-based coordinates of the face $f_R$ of $\calP$ containing $R$. All these
$O(n^2)$ coordinate transformations can be computed in $O(n^2)$
time by a depth-first traversal of the sequence tree from $s$ constructed by the algorithm of Chen and Han~\cite{ref:ChenSh90}. 
Finally, we construct a point location data structure on the subdivision of
$f\setminus \Phi_R$ for each face $f$ of $\calP$~\cite{ref:EdelsbrunnerOp86,ref:KirkpatrickOp83}. All these can be done in $O(n^2)$
time and space for $R$. As there are $O(n^4)$ ridge-free regions, the total preprocessing time and space for this case is $O(n^6)$. 

For two query points $s$ and $t$, suppose that $R$ is the ridge-free region
that contains $s$ and $f$ is the face of $\calP$ that contains $t$. Using the
point location data structure on $f\setminus \Phi_R$, we locate the connected component $C$ of
$f\setminus \Phi_R$ that contains $t$ in $O(\log n)$ time. Using the
label at $C$, we know if the image of $C$ in $\bigstar_s$ is in $\calK_R$. If not, we can compute $d(s,t)$ in $O(1)$ time using the
coordinate transformation at $C$. 

In summary, with $O(n^6)$ time and space preprocessing, for each query $(s,t)$ in the outside-kernel case, $d(s,t)$ can be computed in $O(\log n)$ time. 

\paragraph{The inside-kernel case $\boldsymbol{t^*\in \calK_R}$.}
For the case $t^*\in \calK_R$, let $C$ be the connected component of $f\setminus \Phi_R$ containing $t$, where $f$ is the polyhedron face where $t$ lies. 
Since the kernel $\calK_R$ is fixed for all points $s\in R$, each point of $C$ has a fixed image in $\calK_R$ independent of the point $s\in R$. Correspondingly, in the preprocessing, we compute the coordinate transformation from the $f$-based
coordinates to the coordinates of the star unfolding plane containing $\calK_R$. Doing so for all connected components in all faces of $\calP$ can be done in $O(n^2)$ time for each ridge-free region $R$. 
Furthermore, we perform the following preprocessing work for each ridge-free region
$R$. 

As $s$ moves in $R$, the position of each of its images $s_i$, $1\leq i\leq n$, in $\bigstar_s$ is a linear
function of $s$. The distance of $s_i$ and $t^*$ in the unfolding plane of $\bigstar_s$ defines a 4-variate algebraic
function for $s\in R$ and $t\in \calK_R$. Let $\lowenv_R$ denote the lower envelope of these
$n$ functions. Then, $d(s,t)$ is determined by a vertical ray-shooting query on $\lowenv_R$
using the coordinates of $s$ and $t$. Due to the property that the location of
each $s_i$ is a linear function of the location of $s\in R$, by a linearization
technique, the ray-shooting query can be reduced to a vertical ray-shooting
query on the lower envelope of $n$ 8-variate hyperplanes. Such a ray-shooting query can be answered in $O(\log n)$ time after
$O(n^{4+\epsilon})$ time and space preprocessing~\cite{ref:MatousekRay93}. In the following, we refer to the above data structure as the {\em lower envelope data structure}. Hence, the total
preprocessing time and space for all ridge-free regions $R$ is
$O(n^{8+\epsilon})$. 

In summary, with $O(n^{8+\epsilon})$ time and space preprocessing, for each query $(s,t)$ in the inside-kernel case, $d(s,t)$ can be computed in $O(\log n)$ time. 

\paragraph{Non-generic query points.}
The above discusses the situation when $s$ is a generic point. If $s$ is not a generic point, then $s$ lies on the boundary of more than one ridge-free region. In this case, the query may be solved using the data structure of any ridge-free region whose boundary contains $s$; correctness follows by continuity~\cite{ref:AgarwalRa93}. 



\subsection{Our new solution}
\label{sec:newsol}

We now present our new algorithm for handling the inside-kernel case. 
As discussed in Section~\ref{sec:overview}, instead of constructing a lower envelope data structure for each ridge-free region $R$, we will resort to a structure using cuttings. 

Define $E$ to be the set of the edges of the ridge trees of all polyhedron vertices. Recall that each edge of a ridge tree is a line segment on a polyhedron face. Therefore, $E$ is a set of line segments on $\calP$. 

Consider a face $f$ of $\calP$. Let $E_f$ be the set of line segments of $E$ on $f$.  Let $m_f$ be the number of intersections of the segments of $E_f$. It is tempting to construct a cutting on the supporting lines of the segments of $E_f$. However, since each ridge tree has $O(n^2)$ edges, we have $|E|=O(n^3)$. Hence, the number of intersections between the supporting lines of the segments of $E$ is $O(n^6)$. This would lead to a higher time complexity of the algorithm than our target. In contrast, it is known that the number of intersections between the edges of all ridge trees is only $O(n^4)$~\cite{ref:AgarwalSt97}. Therefore, we will instead construct a cutting on the segments of $E_f$ whose size is sensitive to $m_f$, as discussed in Section~\ref{sec:pre}. 


In the following, we first present a preliminary algorithm, which helps the reader understand the basic idea and is the approach presented in the preliminary version of this work~\cite{ref:WangSh25}. 
It uses a single cutting. 
We will then further improve the algorithm by using a hierarchy of $O(1)$ cuttings. 

\subsubsection{A preliminary solution}

We first consider the following subproblem that will be needed in our preliminary algorithm. 

\paragraph{\bf A crucial subproblem.}
Let $B$ be a region of constant size (e.g., a trapezoid) on a polyhedron face. We assume that the number of segments of $E$ intersecting the interior of $B$ is at most $n^{2/3}$ (as will be clear later, this parameter minimizes the overall preprocessing complexity of our preliminary algorithm). 
We consider the following {\em subproblem}: Build a data structure to compute $d(s,t)$ for queries $(s,t)$ with $s\in B$ and $t\in \calP$. With the above AAOS algorithm, we could do the following. The segments of $E$ partition $B$ into $O(n^{4/3})$ ridge-free regions. For each such region, we construct a lower envelope data structure of $O(n^{4+\epsilon})$ space in $O(n^{4+\epsilon})$ preprocessing time. Hence, the total preprocessing time and space is $O(n^{4+\frac{4}{3}+\epsilon})$ and the query time is $O(\log n)$. In the following, we propose a new method that reduces the preprocessing time and space to $O(n^{4+\epsilon})$. 

Let $V_B$ denote the set of polyhedron vertices $v$ such that the interior of $B$ does not intersect the ridge tree $T_v$ of $v$. Let $V'_B$ denote the set of the polyhedron vertices not in $V_B$; by definition, for each $v\in V'_B$, 
$T_v$ intersects the interior of $B$. 

Consider a point $s\in B$ and its star unfolding $\bigstar_s$. For each source image $s_i$ in $\bigstar_s$, as discussed in Section~\ref{sec:pre}, its two adjacent vertices in $\bigstar_s$ are images of two polyhedron vertices.  If $s$ moves in $B$, by definition, it may cross an edge of $T_v$ for some polyhedron vertex $v\in V'_B$ but it never crosses any edge of $T_v$ for any $v\in V_B$. As $s$ moves in $B$, when it crosses an edge $e$ of $T_v$ of a polyhedron vertex $v$, two source images of $\bigstar_s$ adjacent to $v^*$ merge into one and another source image splits into two, both adjacent to $v^*$, where $v^*$ is the image of $v$ in $\bigstar_s$, while all other source images do not change their relative positions in $\bigstar_s$ (i.e., their adjacent vertices do not change; see Figure~\ref{fig:kernelchange})~\cite{ref:CookSh12}. Hence, each ridge tree edge in the interior of $B$ affects at most four source images as $s$ moves in $B$. 

\begin{figure}[t]
\begin{minipage}[t]{\linewidth}
\begin{center}
\includegraphics[totalheight=2.3in]{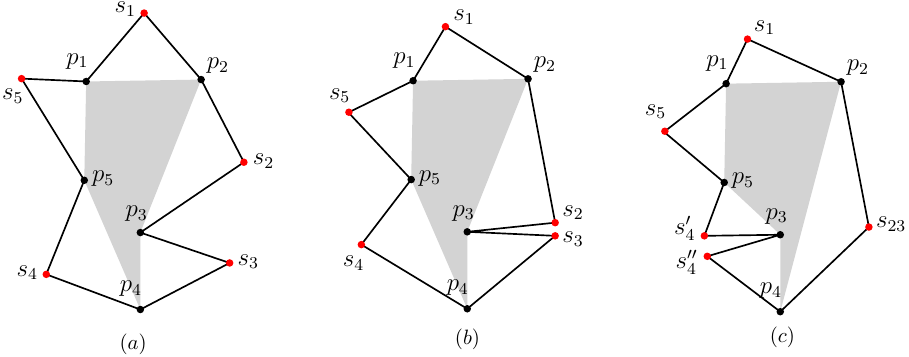}
\caption{
Illustrating the change of $\bigstar_s$ as $s$ moves crossing an edge of $T_{p_3}$: (a) before crossing; (b) about to cross ($s_2$ and $s_3$ are about to merge); and (c) after the crossing ($s_2$ and $s_3$ are merged into a new source image $s_{23}$ while $s_4$ is split into two new source images $s_4'$ and $s_4''$ both connecting to $p_3$).}
\label{fig:kernelchange}
\end{center}
\end{minipage}
\end{figure}


For each source image $s_i$ in $\bigstar_s$, if the adjacent vertices of $s_i$ in $\bigstar_s$ remain the same for all $s\in B$, then the two vertices must be images of two polyhedron vertices from $V_B$ and we call $s_i$ a {\em stable source image} (with respect to $B$); otherwise $s_i$ is {\em unstable} (e.g., in the example of Figure~\ref{fig:star}, suppose that $V_B=\{p_1,p_2,p_3\}$, then only $s_1$ and $s_2$ could possibly be stable).

Since $n^{2/3}$ ridge tree edges intersect the interior of $B$, there are at most $O(n^{2/3})$ unstable source images and at least $n-O(n^{2/3})$ stable source images. We find all stable source images, which can be done in $O(n^{2+\frac{4}{3}})$ time by computing the star unfoldings of all ridge-free regions of $B$. More specifically, there are $O(n^{4/3})$ ridge-free regions in $B$ and the star unfolding of each region can be computed in $O(n^2)$ time~\cite{ref:AgarwalSt97,ref:ChenSh90}.\footnote{In fact, a more efficient algorithm can compute all these star unfoldings in $O(n^2+n\cdot n^{4/3}+n^3)$ time by first computing the star unfolding for one ridge-free region and then update it in $O(n)$ time for each of its neighboring regions~\cite{ref:CookSh12}. We could even do so in a much easier way as here we only need to know the cyclic order of the vertices on the boundary of each source unfolding. The $n^3$ term in the above time complexity is for computing ridge trees of all polyhedron vertices.} Let $S_B$ denote the set of all stable source images.

Pick an arbitrary point $s\in B$ and let $\Pi_B$ be the unfolding plane of $\bigstar_s$. 
We construct a lower envelope data structure for the source images of $S_B$ with respect to $s\in B$ and any point $t^*$ in the plane $\Pi_B$. When $s$ moves in $B$, the plane $\Pi_B$ is fixed and the position of each source image of $S_B$ on $\Pi_B$ is a linear function of the position of $s\in B$~\cite{ref:AgarwalSt97}. As in the AAOS algorithm, constructing the data structure takes $O(n^{4+\epsilon})$ space and preprocessing time. Using the data structure, given any pair of points $(s,t^*)$ with $s\in B$ and $t^*\in \Pi_B$, the point of $S_B$ closest to $t^*$ can be found in $O(\log n)$ time. 

Next, for each ridge-free region $R$ of $B$, consider its star unfolding $\bigstar_{R}$. As $R\subseteq B$, by definition, each point of $S_B$ is also a source image in $\bigstar_R$. Let $S_{R}'$ denote the set of source images in $\bigstar_{R}$ that are not in $S_B$. Hence, $|S_B|+|S_{R}'|=n$. As $|S_B|=n-O(n^{2/3})$, we obtain $|S_{R}'|=O(n^{2/3})$. For convenience, we embed  $\bigstar_{R}$ in the plane $\Pi_B$. We build a lower envelope data structure for the source images of $S_R'$ with respect to $s\in R$ and any point $t^*$ in the plane $\Pi_B$. As above, this takes $O(|S_R'|^{4+\epsilon})$ preprocessing time and space, which is $O(n^{8/3+\epsilon})$ as $|S_R'|=O(n^{2/3})$. Using the data structure, given any pair of points $(s,t^*)$ with $s\in R$ and $t^*\in \Pi_B$, the point of $S_R'$ closest to $t^*$ can be found in $O(\log n)$ time. We do this for all ridge-free regions $R$ of $B$. As $B$ has $O(n^{4/3})$ ridge-free regions, the total space and preprocessing time is $O(n^{4+\epsilon})$. 

In addition, we build a point location data structure on the $O(n^{4/3})$ ridge-free regions of $B$ in $O(n^{4/3})$ time~\cite{ref:EdelsbrunnerOp86,ref:KirkpatrickOp83}. Finally, as in the AAOS algorithm, for each ridge-free region $R$ of $B$, we compute the kernel $\calK_R$ of its star unfolding $\bigstar_R$ embedded in $\Pi_B$, and obtain its preimage $\Phi_R$ on $\calP$. 
Then, for each polyhedron face $f$, for each component $C$ of $f\setminus \Phi_R$, we label it whether its image is in $\calK_R$. If not, then we do the same proprocessing work as in the outside-kernel case of the AAOS algorithm. 
If $C$ is in $\calK_R$, then we compute the coordinate transformation from the $f$-based coordinates to the coordinates of $\Pi_B$. As in the AAOS algorithm, doing the above for each ridge-free region $R$ of $B$ takes $O(n^2)$ time, and takes $O(n^{2+\frac{4}{3}})$ time for all $O(n^{4/3})$ ridge-free regions of $B$.

This finishes the preprocessing for $B$, which takes $O(n^{4+\epsilon})$ space and preprocessing time in total. 


Consider a query $(s,t)$ with $s\in B$ and $t\in \calP$. Using a point location query, we find the ridge-free region $R$ of $B$ that contains $s$ in $O(\log n)$ time. Let $f$ be the face of $\calP$ that contains $t$.  Using the
point location data structure on $f\setminus \Phi_R$, we locate the connected component $C$ of
$f\setminus \Phi_R$ that contains $t$. From the
label at $C$, we know if the image of $C$ in $\bigstar_s$ is in $\calK_R$. If not, we can compute $d(s,t)$ in $O(1)$ time using the
coordinate transformation at $C$ as in the AAOS algorithm. 

If the image of $C$ is in $\calK_R$, then using the coordinate transformation at $C$, we compute the coordinates of $t^*$ in $\Pi_B$. Using the lower envelope data structure for $S_B$, we find the source image $s_i\in S_B$ closest to $t^*$. Next, using the lower envelope data structure for $S_{R}'$, we find the source image $s_j\in S_{R}'$ closest to $t^*$. We then choose the one of $s_i$ and $s_j$ closer to $t^*$, and return their distance as $d(s,t)$. The total query time is $O(\log n)$. As discussed above, the correctness is based on the fact that the set of source images of $\bigstar_s$ is exactly $S_B\cup S_R'$. 

\paragraph{\bf The original problem.}
We are now ready to solve our original problem by using our algorithm for the above subproblem. 
Consider a face $f$ of $\calP$. Let $n_f=|E_f|$. Recall that $m_f$ is the number of intersections of the line segments of $E_f$. Note that $\sum_{f\in \calP}n_f=O(n^3)$ and $\sum_{f\in \calP}m_f=O(n^4)$~\cite{ref:AgarwalSt97}. 

If $n_f\leq n^{2/3}$, then we build the above subproblem data structure for $f$ and $E_f$ (i.e., treat the entire $f$ as $B$), which takes $O(n^{4+\epsilon})$ time and space. Let $\calD_f$ denote this data structure. Using $\calD_f$, $d(s,t)$ for each query $(s,t)$ with $s\in f$ can be computed in $O(\log n)$ time. As $\calP$ has $O(n)$ faces, the total preprocessing time and space for all such faces $f$ of $\calP$ with $n_f\leq n^{2/3}$ is $O(n^{5+\epsilon})$. 


If $n_f>n^{2/3}$, then we resort to cuttings. We compute a $(1/r)$-cutting $\Xi$ for the line segments of $E_f$ with a parameter $r$ such that $n_f/r=n^{2/3}$. As discussed in Section~\ref{sec:pre},  such a cutting of size $O(r^{1+\epsilon}+m_f\cdot r^2/n_f^2)$ exists and can be computed in $O(n_f\cdot r)$ time~\cite{ref:WangUn23}. For each cell $\sigma\in \Xi$, which is a trapezoid, the number of segments of $E_f$ intersecting the interior of $\sigma$ is at most $n_f/r=n^{2/3}$. Hence, it takes $O(n^{4+\epsilon})$ space and preprocessing time to construct the above subproblem data structure on $\sigma$; let $\calD_{\sigma}$ denote the data structure. Using $\calD_{\sigma}$, $d(s,t)$ for each query $(s,t)$ with $s\in \sigma$ can be computed in $O(\log n)$ time. 
If we do this for all polyhedron faces, the total number of cells in the cuttings of all faces is on the order of 
\begin{align*}
    &\sum_{f\in \calP,n_f>n^{2/3}} \left(r^{1+\epsilon}+m_f\cdot r^2/n_f^2\right)     \leq \sum_{f\in \calP}\left(\frac{n_f}{n^{2/3}}\right)^{1+\epsilon} + \sum_{f\in \calP} m_f/n^{4/3} \\
    &\leq n^{\epsilon} \cdot \sum_{f\in \calP}n_f/n^{2/3} + O(n^4/n^{4/3}) =O(n^{3-2/3+\epsilon}+ n^4/n^{4/3}) = O(n^{2+\frac{2}{3}}).
\end{align*}
Therefore, constructing $\calD_{\sigma}$ for all cells $\sigma$ in all cuttings takes $O(n^{6+\frac{2}{3}+\epsilon})$ time and space.
Note that following the above analysis, we have $\sum_{f\in \calP,n_f>n^{2/3}} O(n_f\cdot r)= O(n^3)\cdot \sum_{f\in \calP,n_f>n^{2/3}}r=O(n^3\cdot n^{3-2/3})=O(n^{6-2/3})$, and thus the total time for constructing all the cuttings is $O(n^{6-2/3})$. 
Hence, the total preprocessing time and space is bounded by $O(n^{6+\frac{2}{3}+\epsilon})$.

For each query $(s,t)$, let $f$ be the face of $\calP$ that contains $s$. If $n_f\leq n^{2/3}$, then using the data structure $\calD_f$, $d(s,t)$ can be computed in $O(\log n)$ time. Otherwise, let $\Xi$ be the cutting on $f$. Using $\Xi$, we can find the cell $\sigma\in \Xi$ containing $s$ in $O(\log n)$ time~\cite{ref:ChazelleCu93,ref:WangUn23}. Then, using the data structure $\calD_{\sigma}$, $d(s,t)$ can be computed in $O(\log n)$ time. 

\subsubsection{An improved solution}
\label{sec:improved}

We now improve the above solution. The main idea is to build multi-level cuttings instead of only one. We again start with a subproblem. 

\paragraph{\bf A crucial subproblem.}
Let $B$ be a region of constant size on a face of $\calP$ such that $|E_B|\leq n$, where $E_B$ is set of segments of $E$ intersect the interior of $B$. We show that with $O(n^{4+\epsilon})$ time and space preprocessing, $d(s,t)$ can be computed in $O(\log n)$ time for any query $(s,t)$ with $s\in B$ and $t\in \calP$ (in contrast, the AAOS algorithm would need $O(n^{6+\epsilon})$ preprocessing time and space). 

To simplify the notation, let $n=|E_b|$. As before, let $\Pi_B$ be an unfolding plane of $\bigstar_s$ for $s\in B$. We build a hierarchy of $O(1)$ cuttings on $B$ for $E_B$, as follows. 

We compute a $(1/r)$-cutting $\Xi_1$ for the line segments of $E_B$ with $r=n^{\epsilon}$. Such a cutting of size $O(r^2)$ exists and can be computed in $O(nr)$ time~\cite{ref:ChazelleCu93}. For each cell $\sigma\in \Xi_1$, we use $E_{\sigma}$ to denote the set of segments of $E$ intersecting the interior of $\sigma$. By definition, $|E_{\sigma}|\leq n/r=n^{1-\epsilon}$. With respect to $\sigma$ (and $E_{\sigma}$), we define {\em stable sources images} of the star unfolding $\bigstar_s$ when $s$ moves in $\sigma$ in the same way as before, and let $S_\sigma$ denote the set of stable source images. For each point $s\in \sigma$, as $|E_{\sigma}|\leq n^{1-\epsilon}$, the number of unstable source images is at most $O(n^{1-\epsilon})$. We compute $S_{\sigma}$, which can be done in $O(n^{4-2\epsilon})$ time (e.g., by first computing the arrangement of the segments of $E_{\sigma}$ inside $\sigma$ and then compute the star unfolding $\bigstar_s$ for $s$ in each cell of the arrangement). We further construct a lower envelope data structure $\lowenv_{\sigma}$ for the source images of $S_\sigma$ with respect to $s\in \sigma$ and $t^*\in \Pi_B$, which takes $O(n^{4+\epsilon})$ time and space as $|S_{\sigma}|\leq n$. Computing $S_{\sigma}$ and $\lowenv_{\sigma}$ for all $\sigma\in \Xi_1$ takes $O(n^{4+3\epsilon})$ time and space in total since $\Xi_1$ has $O(n^{2\epsilon})$ cells. 

Next, for each cell $\sigma\in \Xi_1$, we construct an $(1/r)$-cutting $\Xi_\sigma$ for $E_{\sigma}$ inside $\sigma$ with $r=n^{\epsilon}$. Let $\Xi_2$ be the union of all cuttings $\Xi_{\sigma}$ for all $\sigma\in \Xi_1$. Constructing $\Xi_{\sigma}$ takes $O(|E_{\sigma}|\cdot r)=O(n^{1-\epsilon}\cdot n^{\epsilon})=O(n)$ time~\cite{ref:ChazelleCu93}. Hence, computing $\Xi_2$ takes $O(n^{1+2\epsilon})$ time since $\Xi_1$ has $O(n^{2\epsilon})$ cells. Note that $\Xi_2$ has $O(n^{4\epsilon})$ cells and $|E_{\sigma}|\leq n^{1-2\epsilon}$ for each cell $\sigma\in \Xi_2$. Therefore, $\Xi_2$ is essentially a $(1/n^{2\epsilon})$-cutting for $E_B$. Furthermore, each cell of $\Xi_2$ is completely contained in a single cell of $\Xi_1$.

In general, suppose that we already have the cutting $\Xi_i$ that has $O(n^{2i\epsilon})$ cells with $|E_{\sigma}|\leq n^{1-i\epsilon}$ for each cell $\sigma\in \Xi_i$. For each cell $\sigma\in \Xi_i$, we construct a $(1/r)$-cutting $\Xi_\sigma$ for $E_{\sigma}$ inside $\sigma$ with $r=n^{\epsilon}$. Let $\Xi_{i+1}$ be the union of all cuttings $\Xi_{\sigma}$ for all $\sigma\in \Xi_i$. Constructing $\Xi_{\sigma}$ takes $O(|E_{\sigma}|\cdot r)=O(n^{1-i\epsilon}\cdot n^{\epsilon})=O(n^{1-i\epsilon+\epsilon})$ time. Hence, computing $\Xi_{i+1}$ takes $O(n^{1+i\epsilon+\epsilon})$ time since $\Xi_i$ has $O(n^{2i\epsilon})$ cells. Note that $\Xi_{i+1}$ has $O(n^{2(i+1)\epsilon})$ cells and $|E_{\sigma}|\leq n^{1-(i+1)\epsilon}$ for each cell $\sigma\in \Xi_{i+1}$. Therefore, $\Xi_{i+1}$ is essentially a $(1/n^{(i+1)\epsilon})$-cutting for $E_B$. Furthermore, each cell $\sigma$ of $\Xi_{i+1}$ is completely contained in a single cell $\sigma'$ of $\Xi_i$, and we call $\sigma'$ the {\em parent cell} of $\sigma$ (as such, the cells of all cuttings $\Xi_1,\cdots,\Xi_{i+1}$ form a tree structure). 
With respect to $\sigma$ (and $E_{\sigma}$), we define {\em stable sources images} of the star unfolding $\bigstar_s$ when $s$ moves in $\sigma$, and let $S_\sigma$ denote the set of stable source images. In addition, define $S'_{\sigma}=S_{\sigma}\setminus S_{\sigma'}$, where $\sigma'$ is the parent cell of $\sigma$. In particular, we let $S'_{\sigma}=S_{\sigma}$ for all cells $\sigma\in \Xi_1$. 
The following lemma is crucial to the success of the our approach. 

\begin{lemma}\label{lem:multicut}
For any cell $\sigma\in \Xi_{i+1}$, we have the following.
\begin{enumerate}
    \item $S_{\sigma}$ is the union of $S'_{\sigma}$ and $S'_{\sigma'}$ of the cell $\sigma'\in \Xi_j$ containing $\sigma$ for all $1\leq j\leq i$, i.e., if we consider $\sigma$ an ancestor cell of itself, then $S_{\sigma}=\bigcup_{\sigma'\in \Xi_j \text{ is an ancestor of }\sigma,\  1\leq j\leq i+1}S'_{\sigma'}$. 
    \item $|S'_{\sigma}|=O(n^{1-i\epsilon})$. 
\end{enumerate}
\end{lemma}
\begin{proof}
For each $\Xi_j$, $1\leq j\leq i$, let $\sigma_j$ denote the cell of $\Xi_j$ containing $\sigma$. We let $\sigma_{i+1}$ be $\sigma$ itself. 
For the first lemma statement, our goal is to show that $S_{\sigma_{i+1}}=\bigcup_{j=1}^{i+1}S'_{\sigma_j}$. We assume inductively that $S_{\sigma_i}=\bigcup_{j=1}^{i}S'_{\sigma_j}$. The base case holds since $S_{\sigma_1}=S'_{\sigma_1}$ by definition. Hence, our goal is to prove that $S_{\sigma_{i+1}}=S'_{\sigma_{i+1}}\cup S_{\sigma_i}$. By definition, $S'_{\sigma_{i+1}}=S_{\sigma_{i+1}}\setminus S_{\sigma_i}$. Therefore, it suffices to argue that $S_{\sigma_i}\subseteq S_{\sigma_{i+1}}$. 

Consider a source image $s_g\in S_{\sigma_i}$. Let $V_{\sigma_i}$ denote the set of vertices $v$ of $\calP$ such that $T_v$ does not have a segment in $E_{\sigma_i}$. 
By definition, $s_g$ in $\bigstar_s$ is always adjacent to the same two vertices of $V_{\sigma_i}$ for all $s\in \sigma_{i}$. 
We define $V_{\sigma_{i+1}}$ similarly. As $\sigma_{i+1}\subseteq \sigma_i$, $E_{\sigma_{i+1}}\subseteq E_{\sigma_i}$, and this implies that $V_{\sigma_i}\subseteq V_{\sigma_{i+1}}$. Hence, $s_g$ in $\bigstar_s$ is always adjacent to the same two vertices of $V_{\sigma_{i+1}}$ for all $s\in \sigma_{i+1}\subseteq \sigma_i$. Therefore, $s_g\in S_{\sigma_{i+1}}$. 
This proves $S_{\sigma_i}\subseteq S_{\sigma_{i+1}}$. The first lemma statement thus follows. 

We now prove the bound $|S'_{\sigma_{i+1}}|=O(n^{1-i\epsilon})$. Consider a source image $s_g\in S'_{\sigma_{i+1}}$. By definition, $s_g\in S_{\sigma_{i+1}}$ but $s_g\not\in S_{\sigma_i}$. This means that $s_g$ in $\bigstar_s$ is always adjacent to the same two vertices of $V_{\sigma_{i+1}}$ for all $s\in \sigma_{i+1}$, but when $s$ moves in $\sigma_i$, $s$ will cross a segment of the ridge tree $T_v$ of at least one vertex $v\in V_{\sigma_{i+1}}$ whose image is adjacent to $s_g$ in $\bigstar_s$ when $s\in \sigma_{i+1}$. As $|E_{\sigma_i}|\leq n^{1-i\epsilon}$, the number of such source images $s_g$ is at most $O(n^{1-i\epsilon})$. We thus conclude that $|S'_{\sigma_{i+1}}|=O(n^{1-i\epsilon})$. 
\end{proof}

In light of Lemma~\ref{lem:multicut}, we construct a lower envelope data structure $\lowenv_{\sigma}$ for the source images of $S'_{\sigma}$ with respect to $s\in \sigma$ and $t^*\in \Pi_B$, which takes $O(n^{4(1-i\epsilon)+\epsilon})$ time and space as $|S'_{\sigma}|=O(n^{1-i\epsilon})$ by Lemma~\ref{lem:multicut}(2). Doing this for all cells of $\Xi_{i+1}$ takes $O(n^{4-2i\epsilon+3\epsilon})$ time and space, since $\Xi_{i+1}$ has $O(n^{2(i+1)\epsilon})$ cells. This finishes the preprocessing for $\Xi_{i+1}$. 

We continue this until we have computed a cutting $\Xi_k$ such that $n^{1-k\epsilon}\leq n^{\epsilon}$. For simplicity, we choose $\epsilon$ so that $(k+1)\epsilon=1$, making $k=O(1)$. Note that $\Xi_k$ has $O(n^{2k\epsilon})$ cells and $|E_\sigma|\leq n^{1-k\epsilon}\leq n^{\epsilon}$ for all cells $\sigma\in \Xi_k$. The total preprocessing time and space for all cuttings $\Xi_1,\Xi_2,\ldots,\Xi_k$ is on the order of $\sum_{i=0}^{k-1}n^{4-2i\epsilon+3\epsilon}$, which is $O(n^{4+3\epsilon})$. Note that these cuttings can also be obtained by first computing a hierarchical sequence of cuttings for $E_B$ using Chazelle's algorithm~\cite{ref:ChazelleCu93} and then picking above cuttings from the sequence. 

Next, we perform the following additional preprocessing for each cell $\sigma$ in the last cutting $\Xi_k$. 
We compute the arrangement $A_{\sigma}$ of the segments of $E_{\sigma}$ inside $\sigma$.
For each cell $R\in A_{\sigma}$, consider its star unfolding $\bigstar_{R}$ (note that the interior of $R$ does not intersect any segment of $E$). As $R\subseteq \sigma$, 
by definition, each point of $S_\sigma$ is also a source image in $\bigstar_R$. Let $S_{R}$ denote the set of source images in $\bigstar_{R}$ and let $S_R'=S_R\setminus S_\sigma$. By Lemma~\ref{lem:multicut}, $S_R$ is the union of $S_R'$ and $S'_{\sigma'}$ for all ancestor cells $\sigma'\in \Xi_i$ of $\sigma$ for all $1\leq i\leq k$ (we consider $\sigma$ an ancestor of itself).
Since $|E_\sigma|\leq n^{\epsilon}$, by the same argument as in the proof of Lemma~\ref{lem:multicut}(2), $|S_R'|=O(n^{\epsilon})$. We build a lower envelope data structure $\lowenv_R$ for the source images of $S_R'$ with respect to $s\in R$ and $t^*\in \Pi_B$. This takes $O(|S_R'|^{4+\epsilon})$ preprocessing time and space, which is $O(n^{4\epsilon+\epsilon^2})$ as $|S_R'|=O(n^{\epsilon})$. 
Given any query $(s,t^*)$ with $s\in R$ and $t^*\in \Pi_B$, the point of $S_R$ closest to $t^*$ is the point of $S_R'$ closest to $t^*$ or the point of $S'_{\sigma'}$ closest to $t^*$ for some ancestor cell $\sigma'\in \Xi_i$ of $\sigma$, $1\leq i\leq k$. As $k=O(1)$, the point of $S_R$ closest to $t^*$ can be computed in $O(\log n)$ time using the lower envelope data structures $\lowenv_R$ and $\lowenv_{\sigma'}$ for all ancestor cells $\sigma'\in \Xi_i$ of $\sigma$ for all $1\leq i\leq k$. 
Since $|E_{\sigma}|=O(n^{\epsilon})$, $|A_{\sigma}|=O(n^{2\epsilon})$. 
As $\Xi_k$ has $O(n^{2k\epsilon})$ cells, preprocessing all cells $\sigma\in \Xi_k$ as above takes $O(n^{2k\epsilon+2\epsilon +4\epsilon+\epsilon^2})$ time and space, which is $O(n^{2+4\epsilon+\epsilon^2})$. 

For convenience, we use $\Xi_{k+1}$ to denote the collection of all arrangement cells $R\in A_{\sigma}$ for all $\sigma\in \Xi_k$. Hence, the number of cells $R$ in $\Xi_{k+1}$ is $O(n^{2k\epsilon+2\epsilon})=O(n^2)$. We construct a point location data structure on the cells of $\Xi_{k+1}$ in $O(n^2)$ time.

In addition, as before, for each cell $R$ of $\Xi_{k+1}$, we compute the kernel $\calK_R$ of its star unfolding $\bigstar_R$ embedded in $\Pi_B$, and obtain its preimage $\Phi_R$ on $\calP$. 
Then, for each polyhedron face $f$, for each component $C$ of $f\setminus \Phi_R$, we label it whether its image is in $\calK_R$. If not, then we do the same proprocessing work as in the outside-kernel case of the AAOS algorithm. 
If $C$ is in $\calK_R$, then we compute the coordinate transformation from the $f$-based coordinates to the coordinates of $\Pi_B$. As before, doing the above for each cell $R$ of $\Xi_{k+1}$ takes $O(n^2)$ time, and the total time over all $O(n^2)$ cells of $\Xi_{k+1}$ is $O(n^4)$.

This finishes the preprocessing for $B$, which takes $O(n^{4+3\epsilon})$ time and space in total. 


Consider a query $(s,t)$ with $s\in B$ and $t\in \calP$. Using a point location query, we find the cell $R$ of $\Xi_{i+1}$ that contains $s$ in $O(\log n)$ time. Let $f$ be the face of $\calP$ that contains $t$.  Using the
point location data structure on $f\setminus \Phi_R$, we locate the connected component $C$ of
$f\setminus \Phi_R$ that contains $t$. From the
label at $C$, we know whether the image of $C$ in $\bigstar_s$ is in $\calK_R$. If not, we can compute $d(s,t)$ in $O(1)$ time using the
coordinate transformation at $C$. 

If the image of $C$ is in $\calK_R$, then using the coordinate transformation at $C$, we compute the coordinates of $t^*$ in $\Pi_B$. As discussed above, using the lower envelope data structures, we can find the source image $s_i\in S_{R}$ closest to $t^*$ in $O(\log n)$ time, and return their distance as $d(s,t)$. The total query time is $O(\log n)$. 

\paragraph{\bf The original problem.}
We now solve the original problem using our algorithm for the above subproblem. 

For each face $f$ of $P$. If $n_f\leq n$, then we build the above data structure for $f$ and $E_f$, which takes $O(n^{4+3\epsilon})$ time and space. Let $\calD_f$ denote this data structure. Using $\calD_f$, $d(s,t)$ for each query $(s,t)$ with $s\in f$ can be computed in $O(\log n)$ time. As $\calP$ has $O(n)$ faces, the total preprocessing time and space for all such faces $f$ of $\calP$ with $n_f\leq n$ is $O(n^{5+3\epsilon})$. 

If $n_f>n$, 
we compute a $(1/r)$-cutting $\Xi$ for $E_f$ with a parameter $r$ such that $n_f/r=n$. As discussed in Section~\ref{sec:pre},  such a cutting of size $O(r^{1+\epsilon}+m_f\cdot r^2/n_f^2)$ exists and can be computed in $O(n_f\cdot r)$ time~\cite{ref:WangUn23}. For each cell $\sigma\in \Xi$, which is a trapezoid, the number of segments of $E_f$ intersecting the interior of $\sigma$ is at most $n_f/r=n$. 
Hence, it takes $O(n^{4+3\epsilon})$ space and preprocessing time to construct the above subproblem data structure on $\sigma$; let $\calD_{\sigma}$ denote the data structure. Using $\calD_{\sigma}$, $d(s,t)$ for each query $(s,t)$ with $s\in \sigma$ can be computed in $O(\log n)$ time. 
If we do this for all polyhedron faces, the total number of cells in the cuttings of all faces is on the order of 
\begin{align*}
    &\sum_{f\in \calP,n_f>n} \left(r^{1+\epsilon}+m_f\cdot r^2/n_f^2\right)     \leq \sum_{f\in \calP}\left(\frac{n_f}{n}\right)^{1+\epsilon} + \sum_{f\in \calP} m_f/n^2 \\
    &\leq n^{\epsilon} \cdot \sum_{f\in \calP}n_f/n + O(n^4/n^2) =O(n^{3-1+\epsilon}+ n^2) = O(n^{2+\epsilon}).
\end{align*}

Therefore, constructing $\calD_{\sigma}$ for all cells $\sigma$ in all cuttings takes $O(n^{6+4\epsilon})$ time and space.
The total time of constructing the cuttings for all faces is $\sum_{f\in \calP,n_f>n} O(n_f\cdot r)= O(n^3)\cdot \sum_{f\in \calP,n_f>n}r=O(n^3\cdot n^{3-1})=O(n^5)$.
Hence, the total preprocessing time and space is bounded by $O(n^{6+4\epsilon})$, which can be written as $O(n^{6+\epsilon})$ for a bigger $\epsilon$.

For each query $(s,t)$, let $f$ be the face of $\calP$ that contains $s$. If $n_f\leq n$, then using the data structure $\calD_f$, $d(s,t)$ can be computed in $O(\log n)$ time. Otherwise, let $\Xi$ be the cutting on $f$. Using $\Xi$, we can find the cell $\sigma\in \Xi$ containing $s$ in $O(\log n)$ time~\cite{ref:ChazelleCu93,ref:WangUn23}. Then, using the data structure $\calD_{\sigma}$, $d(s,t)$ can be computed in $O(\log n)$ time.

\subsubsection{Reporting a shortest path $\boldsymbol{\pi(s,t)}$}
The above focuses on computing the geodesic distance $d(s,t)$. How to report $\pi(s,t)$ was not discussed in the AAOS algorithm. In the following, we show that with some extra work (without affecting the overall preprocessing complexities asymptotically), $\pi(s,t)$ can be reported in $O(\log n+|\pi(s,t)|)$ time. 

Let $\calR$ denote the set of the ridge-free regions of $\calP$ subdivided by the ridge trees $T_v$ of all vertices $v$ of $\calP$. Recall that $|\calR|=O(n^4)$ and $\calR$ can be computed in $O(n^4)$ time~\cite{ref:AgarwalSt97}. We construct a point location data structure on the ridge-free regions in each face of $\calP$. This takes $O(n^4)$ time and space in total for all faces. Given a query point $s$ on a face of $\calP$, the ridge-free region $R\in \calR$ containing $s$ can be found in $O(\log n)$ time by a point location query. 

\paragraph{The outside-kernel case $\boldsymbol{t^*\not\in \calK_R}$.}
For the outside-kernel case $t^*\not\in \calK_R$, where $R\in \calR$ is the ridge-free region containing $s$ and $t^*$ is the image of $t$ in the star unfolding $\bigstar_s$, we perform the following additional work in the preprocessing. We follow the notation in Section~\ref{sec:review}. 
For each ridge-free region $R$, using Chen and Han's algorithm~\cite{ref:ChenSh90}, we compute a sequence tree in $O(n^2)$ time and space that stores shortest path edge sequences from $p_R$ to all $O(n^2)$ representative points $q_C$, so that for any $q_C$, the edge sequence of $\pi(p_R,q_C)$ can be output in time linear in the number of edges in the edge sequence. As there are $O(n^4)$ ridge-free regions, the total time and space for doing this for all ridge-free regions is $O(n^6)$. 

Given a query $(s,t)$ with $s\in R$, $t\in \calP$, and $t^*\not\in \calK_R$, suppose that $C$ is the connected component of $f\setminus \Phi_R$ that contains $t$, where $f$ is the polyhedron face containing $t$. Then, the edge sequence of $\pi(s,t)$ is the same as the edge sequence of $\pi(p_R,q_C)$, which can be found in linear time using the edge sequence tree for $R$. 
Consequently, $\pi(s,t)$ can be obtained in $O(|\pi(s,t)|)$ time by unfolding the faces along the edge sequence. 

\paragraph{The inside-kernel case $\boldsymbol{t^*\in \calK_R}$.}
In what follows, we discuss how to report $\pi(s,t)$ in the insider-kernel case. We will resort to pasting trees as defined below. 

For a ridge-free region $R\in \calR$, let $\Upsilon_R$ denote the dual graph of $\bigstar_R$, i.e., its nodes are the plates of $\bigstar_R$ and two nodes have an arc if their corresponding plates share an edge of $\bigstar_R$ (which is the image of a segment of a polyhedron edge); it is known that $\Upsilon_R$ is a tree, called the {\em pasting tree}~\cite{ref:AgarwalSt97}. The pasting tree has $n$ leaves corresponding to the $n$ triangular plates incident to the $n$ source images. Since $\bigstar_R$ has $O(n^2)$ edges and nodes, $\Upsilon_R$ also has $O(n^2)$ arcs and nodes. For any $s\in R$ and $t\in \calP$, $\pi(s,t)$ corresponds to a simple path in $\Upsilon_R$. More specifically, let $t^*$ be the image of $t$ in the star unfolding $\bigstar_s$. We assume $t^*\in \calK_R$ since we are considering the inside-kernel case. The kernel $\calK_R$ is subdivided into cells by the images of the polyhedron edges and this subdivision is independent of $s\in R$. Each cell of the subdivision of $\calK_R$ is contained in a plate of $\bigstar_R$ that thus corresponds to a node of $\Upsilon_R$. If the image of $\pi(s,t)$ is the line segment $\overline{s_it^*}$ in $\bigstar_s$, where $s_i$ is a source image, then the edge sequence of $\pi(s,t)$ corresponding to the arcs of the path in $\Upsilon_R$ from $v_s$ to $v_t$, where $v_s$ is the leaf of $\Upsilon_R$ corresponding to the plate incident to $s_i$ and $v_t$ is the node of $\Upsilon_R$ corresponding to the cell of the subdivision of $\calK_R$ containing $t^*$. Therefore, once we know the path from $v_s$ to $v_t$ in $\Upsilon_R$, $\pi(s,t)$ can be output in $O(|\pi(s,t)|)$ time. Hence, given a query $(s,t)$ with $s\in R$ in the inside-kernel case, to report $\pi(s,t)$, we need to know the following information: (1) the source image $s_i$ that gives the geodesic distance $d(s,t)$, i.e., $\overline{s_it^*}$ is the image of $\pi(s,t)$ in $\bigstar_s$; (2) the cell of the subdivision of $\calK_R$ that contains $t^*$. In addition, we need to have a data structure for $\Upsilon_R$ so that given a leaf $v_s$ and a node $v_t$, the path in $\Upsilon_R$ from $v_s$ to $v_t$ can be output in time linear in the number of arcs of the path. To this end, we perform the following additional work in the preprocessing.

First of all, notice that our above query algorithm for computing the geodesic distance $d(s,t)$ also determines the source image $s_i$ of $\bigstar_s$ that gives the length $d(s,t)$. Indeed, the lower envelope data structure query algorithm also returns the function that gives the value $d(s,t)$ and the source image that defines the function is $s_i$. 

For each ridge-free region $R\in \calR$, after computing the kernel $\calK_R$, we construct a point location data structure on the subdivision of $\calK_R$. Since the subdivision has $O(n^2)$ cells, this takes $O(n^2)$ time and space~\cite{ref:EdelsbrunnerOp86,ref:KirkpatrickOp83}. Next, we compute the pasting tree $\Upsilon_R$ in $O(n^2)$ time. We construct a {\em path query} data structure to answer the following queries:
Given two nodes of $\Upsilon_R$, return the path in $\Upsilon_R$ connecting the two nodes. Since $\Upsilon_R$ has $O(n^2)$ nodes, it is
possible to build a data structure in $O(n^2)$ time and space such that each path query can
be answered in linear time in the number of arcs of the path. For example, one can
build a lowest common ancestor data structure in $O(n^2)$ time for $\Upsilon_R$ so that the lowest
common ancestor of two query nodes can be found in $O(1)$ time~\cite{ref:HarelFa84,ref:BenderTh00}. Using the lowest
common ancestor, it is easy to find the path connecting the two query nodes in
linear time in the number of edges of the path. In this way, we spend $O(n^2)$ additional preprocessing time and space for each ridge-free region $R\in \calR$. The total additional preprocessing time and space is thus $O(n^6)$.

For each query $(s,t)$ in the inside-kernel case, suppose that $R\in \calR$ is the ridge-free region
containing $s$ and $s_i$ is the source image giving the geodesic distance
$d(s,t)$. Using a point location query, we can find the cell $\sigma$ of the subdivision of $\calK_R$
that contains $t^*$ (recall that $t^*$ can be obtained after locating the connected component $C$ of $f\setminus \Phi_R$ containing $t$, where $f$ is the polyhedron face containing $t$), and thus the node of $\Upsilon$ corresponding to $\sigma$ is also determined. In this way, the two nodes $v_s$ and $v_t$ of $\Upsilon_R$ are obtained. Then, using the path query data structure for $\Upsilon_R$, $\pi(s,t)$ can be finally reported in $O(|\pi(s,t)|)$ time. 

\medskip
\medskip

The following theorem summarizes our result in this section. 

\begin{theorem}
Given a convex polyhedral surface $\calP$ of $n$ vertices, we can construct a data structure in $O(n^{6+\epsilon})$ preprocessing time and space, such that for any two query points $s$ and $t$ on $\calP$, the geodesic distance $d(s,t)$ can be computed in $O(\log n)$ time and $\pi(s,t)$ can be output in additional $O(|\pi(s,t)|)$ time. 
\end{theorem}


\section{Two-point shortest path queries -- the edge-face case} 
\label{sec:edge}

In this section, we consider the edge-face case of the two-point shortest
path query problem. In this case, one of the query points $s$ and $t$ is
required to be on an edge of $\calP$. Without loss of generality, we assume that
$s$ is on an edge. 
We follow the idea in Section~\ref{sec:twopoint} to use hierarchical cuttings, but 
the method is much easier here since we only need cuttings in the 1D space. 

Recall from Section~\ref{sec:pre} that the ridge trees of all polyhedron vertices partition the edges of $\calP$ into $O(n^3)$
edgelets and can be computed in $O(n^3\log n)$ time~\cite{ref:AgarwalSt97}. For each edgelet
$\eta$, if $s$ changes in $\eta$, its star unfolding $\bigstar_s$ does not
change topologically, so we use $\bigstar_{\eta}$ to denote it and use $\calK_{\eta}$ to denote its kernel, which is fixed for all $s\in \eta$~\cite{ref:AgarwalSt97}. 

Consider a query $(s,t)$ with $s\in \eta\subseteq e$ for an edgelet $\eta$ on a polyhedron edge $e$. Let $t^*$ be the image of $t$ in the star unfolding $\bigstar_s$. As in the general case, depending on whether $t^*$ is in the kernel $\calK_s$ of $\bigstar_s$, there are two cases. Using the same approach as in Section~\ref{sec:review}, we can compute $t^*$ and determine whether $t^*$ is in $\calK_{s}$ in $O(\log n)$ time after performing $O(n^2)$ time preprocessing for $\eta$. 
If $t^*\not\in \calK_s$, then we can use the same method as before to answer each query in $O(\log n)$ time after $O(n^2)$ time preprocessing for $\eta$. Since there are $O(n^3)$ edgelets in total, the total preprocessing time and space for this case is $O(n^5)$. In the following, we discuss the inside-kernel case $t^*\in \calK_s$. 

We follow an approach similar to that in Section~\ref{sec:twopoint}. But since $s$ is now required to be on the polyhedron edges, instead of constructing cuttings in 2D, we construct cuttings on the edges of $\calP$, which are simply partitions of the edges into segments. Specifically, for each polyhedron edge $e$, we partition $e$ into segments such that each segment has exactly $n$ edgelets except that one segment possibly has less than $n$ edgelets. In this way, the edges of $\calP$ are partitioned into $O(n^{2})$ segments such that each segment has no more than $n$ edgelets. In the following, for each such segment $\beta$, we build a data structure in $O(n^{3+\epsilon})$ preprocessing time and space, such that each query $(s,t)$ with $s\in e$ can be answered in $O(\log n)$ time.

\paragraph{Solving a subproblem for $\boldsymbol{\beta}$.} 
For convenience, we assume that $\beta$ has exactly $n$ edgelets. 
Following the method in Section~\ref{sec:improved}, we will construct a hierarchical sequence of cuttings $\Xi_1,\Xi_2,\cdots,\Xi_k$ with $k=O(1)$ for $\beta$ as follows. 

Since $\beta$ has $n$ edgelets, we partition $\beta$ into $O(n^{\epsilon})$ subsegments each containing at most $n^{1-\epsilon}$ edgelets. These subsegments form our first cutting $\Xi_1$. For each subsegment $\sigma\in \Xi_1$ (also called a {\em cell} of $\Xi_1$), as before, we define {\em stable source images} of $\bigstar_s$ with respect to $s\in \sigma$, and 
let $S_{\sigma}$ denote the set of stable source images. 
Since $\sigma$ has at most $n^{1-\epsilon}$ edgeles, $S_{\sigma}$ can be computed in $O(n^{3-\epsilon})$ time. Note that $|S_{\sigma}|\leq n$.
We further construct a lower envelope data structure $\lowenv_{\sigma}$ for the source images of $S_{\sigma}$ for $s\in \sigma$ and $t^*\in \Pi_\beta$, where $\Pi_\beta$ is the unfolding plane of $\calK_{\eta}$ for an arbitrary edgelet $\eta$ on $\beta$ (i.e., the distance each source image of $S_{\sigma}$ and $t^*$ in $\Pi$ defines a $3$-variate algebraic function of constant degree). 
By a linearization method to reduce the problem to ray-shooting queries on the lower envelope of $O(n)$ 6-variate hyperplanes, the lower envelope data structure can be constructed 
in $O(n^{3+\epsilon})$ space and preprocessing time such that each query can be answered in $O(\log n)$ time (i.e., given $(s,t)$ with $s\in \sigma$ and $t\in \Pi_\beta$, the source image of $S_{\sigma}$ closest to $t^*$ can be found in $O(\log n)$ time)~\cite{ref:AgarwalSt97,ref:MatousekRay93}. 
Computing $S_{\sigma}$ and $\lowenv_\sigma$ for all $\sigma\in \Xi_1$ takes $O(n^{3+2\epsilon})$ time and space in total since $\Xi_1$ has $O(n^{\epsilon})$ cells.

Suppose that we already have the cutting $\Xi_i$ that has $O(n^{i\epsilon})$ subsegments each containing 
$n^{1-i\epsilon}$ edgelets. For each cell $\sigma\in \Xi_i$, we partition it into $O(n^{\epsilon})$ subsegments each containing $n^{1-(i+1)\epsilon}$ edgelets. Let $\Xi_{i+1}$ be the collection of subsegments of all cuttings $\Xi_{\sigma}$ for all $\sigma\in \Xi_i$.
Constructing $\Xi_{\sigma}$ can be easily done in $O(n^{1-i\epsilon}\log n)$ time (e.g., by first sorting the edgelet endpoints). 
Hence, computing $\Xi_{i+1}$ takes $O(n\log n)$ time since $\Xi_i$ has $O(n^{i\epsilon})$ cells. Note that $\Xi_{i+1}$ has $O(n^{(i+1)\epsilon})$ cells. Furthermore, each cell $\sigma$ of $\Xi_{i+1}$ is completely contained in a single cell $\sigma'$ of $\Xi_i$, and we call $\sigma'$ the {\em parent cell} of $\sigma$. 
With respect to $\sigma$, we define {\em stable sources images} of the star unfolding $\bigstar_s$ when $s$ moves in $\sigma$, and let $S_\sigma$ denote the set of stable source images. In addition, define $S'_{\sigma}=S_{\sigma}\setminus S_{\sigma'}$, where $\sigma'$ is the parent cell of $\sigma$. In particular, we let $S'_{\sigma}=S_{\sigma}$ for all cells $\sigma\in \Xi_1$. 
Lemma~\ref{lem:multicut} is still applicable here. 
We construct a lower envelope data structure $\lowenv_{\sigma}$ for the source images of $S'_{\sigma}$ with respect to $s\in \sigma$ and $t^*\in \Pi_\beta$, which takes $O(n^{3(1-i\epsilon)+\epsilon})$ time and space as $|S'_{\sigma}|=O(n^{1-i\epsilon})$ by Lemma~\ref{lem:multicut}(2). Doing this for all cells of $\Xi_{i+1}$ takes $O(n^{3-2i\epsilon+2\epsilon})$ time and space, since $\Xi_{i+1}$ has $O(n^{(i+1)\epsilon})$ cells. This finishes the preprocessing for $\Xi_{i+1}$. 

We continue this until we have computed a cutting $\Xi_k$ such that $n^{1-k\epsilon}\leq n^{\epsilon}$. For simplicity, we choose $\epsilon$ so that $(k+1)\epsilon=1$, making $k=O(1)$. Note that $\Xi_k$ has $O(n^{k\epsilon})$ cells each containing at most $n^{\epsilon}$ edgelets. The total preprocessing time and space for all cuttings $\Xi_1,\Xi_2,\ldots,\Xi_k$ is on the order of $\sum_{i=0}^{k-1}n^{3-2i\epsilon+2\epsilon}$, which is $O(n^{3+2\epsilon})$. 

Next, we perform the following additional preprocessing for each cell $\sigma$ in the last cutting $\Xi_k$. 
For each edgelet $\eta$ in $\sigma$, consider its star unfolding $\bigstar_{\eta}$. As $\eta\subseteq \sigma$, 
by definition, each point of $S_\sigma$ is also a source image in $\bigstar_\eta$. Let $S_{\eta}$ denote the set of source images in $\bigstar_{\eta}$ and let $S_\eta'=S_\eta\setminus S_\sigma$. By Lemma~\ref{lem:multicut}, $S_\eta$ is the union of $S_\eta'$ and $S'_{\sigma'}$ for all ancestor cells $\sigma'\in \Xi_i$ of $\sigma$ for all $1\leq i\leq k$.
Since $\sigma$ has at most $n^{\epsilon}$ subsegments, we have $|S_\eta'|=O(n^{\epsilon})$. We build a lower envelope data structure $\lowenv_\eta$ for the source images of $S_\eta'$ with respect to $s\in \eta$ and $t^*\in\Pi_\beta$. This takes $O(|S_\eta'|^{3+\epsilon})$ preprocessing time and space, which is $O(n^{3\epsilon+\epsilon^2})$ as $|S_\eta'|=O(n^{\epsilon})$. 
Given any pair of points $(s,t^*)$ with $s\in \eta$ and $t^*\in \Pi_\beta$, the point of $S_\eta$ closest to $t^*$ is the point of $S_\eta'$ closest to $t^*$ or the point of $S'_{\sigma'}$ closest to $t^*$ for some ancestor cell $\sigma'\in \Xi_i$ of $\sigma$, $1\leq i\leq k$. As $k=O(1)$, the point of $S_\eta$ closest to $t^*$ can be computed in $O(\log n)$ time using the lower envelope data structures $\lowenv_\eta$ and $\lowenv_{\sigma'}$ for all ancestor cells $\sigma'\in \Xi_i$ of $\sigma$ for all $1\leq i\leq k$. 
As $\Xi_k$ has $O(n^{k\epsilon})$ cells each containing at most $n^{\epsilon}$ edgelets, preprocessing all cells $\sigma\in \Xi_k$ as above takes $O(n^{k\epsilon+\epsilon +3\epsilon+\epsilon^2})$ time and space, which is $O(n^{1+3\epsilon+\epsilon^2})$. 

For convenience, we use $\Xi_{k+1}$ to denote the collection of all edgelets on $\beta$. Hence, $|\Xi_{k+1}|\leq n$. We sort the endpoints of all edgelets of $\Xi_{k+1}$ on $\beta$ to facilitate binary search. For each edgelet $\eta\in \Xi_{k+1}$, we also associate with it the cell of $\sigma$ of $\Xi_k$ that contains it. This finishes the preprocessing for $\beta$, which takes $O(n^{3+2\epsilon})$ time and space in total. 

For a query $(s,t)$ with $s\in \beta$ in the inside-kernel case, by binary search we find the edgelet $\eta$ of $\Xi_{k+1}$ that contains $s$. Then, as discussed above, using the lower envelope data structures, we compute the source image of $s_i\in S_{\eta}$ closest to $t^*$ in $O(\log n)$ time (recall that $t^*$ can be computed in $O(\log n)$ time), and return their distance as $d(s,t)$. The total query time is $O(\log n)$.  

\paragraph{Solving the original problem.} 
As the polyhedron edges are partitioned into $O(n^{2})$ segments $\beta$, the total preprocessing time and space is $O(n^{5+2\epsilon})$, which can be written as $O(n^{5+\epsilon})$ for a larger $\epsilon$. For each query $(s,t)$, it takes $O(\log n)$ time to compute the geodesic distance $d(s,t)$. To output the shortest path $\pi(s,t)$, in the same way as in Section~\ref{sec:twopoint} for the general case, we can perform $O(n^2)$ time additional preprocessing work for each edgelet, with a total of $O(n^5)$ additional preprocessing time and space, 
so that $\pi(s,t)$ can be reported in $O(\log n+ |\pi(s,t)|)$ time.


The following theorem summarizes our result. 

\begin{theorem}\label{theo:edgeface}
Given a convex polyhedral surface $\calP$ of $n$ vertices, we can construct a data structure in $O(n^{5+\epsilon})$ preprocessing time and space, such that for any two query points $s$ and $t$ on $\calP$ with $s$ on an edge of $\calP$, the geodesic distance $d(s,t)$ can be computed in $O(\log n)$ time and a shortest path $\pi(s,t)$ can be output in additional $O(|\pi(s,t)|)$ time. 
\end{theorem}

\paragraph{The single-edge case.}
If a query point, say $s$, is required to be on a single edge $e$ of $\calP$, then the preprocessing complexities of Theorem~\ref{theo:edgeface} can be reduced to $O(n^{4+\epsilon})$. Indeed, since each ridge tree can intersect $e$ at most $O(n)$ times, the number of edgelets on $e$ is $O(n^2)$. Hence, $e$ can be partitioned into $O(n)$ segments $\beta$. Therefore, the total preprocessing time and space become $O(n^{4+\epsilon})$. 

\begin{theorem}\label{theo:singleedgeface}
Given a convex polyhedral surface $\calP$ of $n$ vertices and an edge $e$, we can construct a data structure in $O(n^{4+\epsilon})$ preprocessing time and space, such that for any two query points $s\in e$ and $t\in \calP$, the geodesic distance $d(s,t)$ can be computed in $O(\log n)$ time and a shortest path $\pi(s,t)$ can be output in additional $O(|\pi(s,t)|)$ time. 
\end{theorem}


\subsection{A grouping method}
\label{sec:group}

We present another method, called the {\em grouping method}, which can achieve $O(\log^2 n)$ query time with the same preprocessing as Theorem~\ref{theo:edgeface} (the preprocessing time becomes expected). Although the query time is worse, the method has the virtue of simplicity. In addition, the same method will also be used later in Section~\ref{sec:seq}. 

As in Section~\ref{sec:newsol}, our grouping method is motivated by a simple observation: If $\eta$ and $\eta'$ are two adjacent edgelets on the same edge of $\calP$, then their star unfoldings differ by at most four source images~\cite{ref:CookSh12}. 
As discussed in the above cutting approach, each source image $s_i$ in the star unfolding $\bigstar_{\eta}$ of $\eta$ defines a $3$-variate algebraic function $f_i(s,t^*)$ of constant degree for $s\in \eta$ and $t^*$ in the unfolding plane of $\calK_{\eta}$. Let $F(\eta)$ denote the set of functions $f_i$ defined by all $n$ source images. If a source image $s_i$ is still in the star unfolding of $\eta'$, then $s_i$ still defines the same function $f_i$ but with the domain of $s$ changed to $\eta'$. Since $\eta$ and $\eta'$ lie in the same line, the two functions are essentially one function with the domain of $s$ set to $\eta\cup \eta'$ (note that we can have $\calK_{\eta}$ and $\calK_{\eta'}$ lie in the same plane~\cite{ref:CookSh12} so that the domain of $t^*$ is always in that plane). In this way, $F(\eta)\cup F(\eta')$ only have at most $n+4$ functions. 

We propose the following {\em grouping method}. We group every adjacent $n$ edgelets on each edge $e$ of $\calP$ with the exception that the last group may have less than $n$ edgelets. Following the above discussion, the total number of functions defined by all source images in the star unfoldings of all $n$ edgelets in the same group is $O(n)$. Then, we construct a vertical ray-shooting  query data structure for the lower envelope $\lowenv_G$ of all functions defined by all edgelets in each group $G$. Note that this time we cannot reduce the problem to vertical ray-shootings on lower envelopes of hyperplanes anymore because functions may not fully be defined on the segment which is the union of the edglets in each group. Instead, since each function is a $3$-variate algebraic function of constant degree and constant size, we can use the algorithm in \cite{ref:AgarwalCo97} to compute a data structure of $O(n^{3+\epsilon})$ space in $O(n^{3+\epsilon})$ randomized expected time so that each vertical ray-shooting query on $\lowenv_G$ can be answered in $O(\log^2 n)$ time. 

For each query $(s,t)$ in the inside-kernel case with $s$ in the union of the edgelets of $G$, computing $d(s,t)$ can be reduced to a vertical ray-shooting query on $\lowenv_G$, which can then be answered in $O(\log^2 n)$ time. 

As there are $O(n^3)$ edgelets in all edges of $\calP$, there are a total of $O(n^2)$ groups. The total preprocessing time and space is thus bounded by $O(n^{5+\epsilon})$. For each query $(s,t)$, we first determine which group contains $s$ by binary search, and then apply the query algorithm for the group. Therefore, each query can be answered in $O(\log^2 n)$ time. Note that this is for the inside-kernel case queries. For the outside-kernel case queries, we already discussed in the beginning of Section~\ref{sec:edge} that with $O(n^{5})$ time preprocessing, each query can be answered in $O(\log n)$ time. 

\section{Shortest path edge sequences}
\label{sec:seq}

In this section, we present our $O(n^{5+\epsilon})$ time algorithm to compute the exact set of shortest path edge sequences on $\calP$. Instead of using star unfolding as in \cite{ref:AgarwalSt97} or using kinetic Voronoi diagrams as in~\cite{ref:GuibasVo91}, we propose a new direction to tackle the problem.  We begin with an overview of our approach. 

\subsection{Overview}
Let $s$ be a point on an edge $e$ of $\calP$. Define $\Sigma_s$ as the set of edge sequences of $\pi(s,t)$ for all points $t\in \calP$. Our first observation is that the set $\Sigma_s$ is uniquely determined by the ridge tree $T_s$ of $s$. Recall that a vertex $v$ of $T_s$ belongs to the following three types: (1) $v$ is a leaf, which is of degree one and is a vertex of $\calP$; (2) $v$ is on a polyhedron edge and is of degree two; (3) $v$ is a high-degree vertex (i.e., of degree at least $3$), which has at least three shortest paths from $s$. Due to our general position assumption, each vertex of $T_s$ has degree at most $4$, and has degree at most $3$ if the vertex is on a polyhedron edge. 

Define $\Sigma(e)$ to be the set of edge sequences of $\pi(s,t)$ for all points $s\in e$ and $t\in \calP$. Our goal is to compute $\Sigma(e)$ for all polyhedron edges $e$. We will describe our algorithm to compute $\Sigma(e)$ for a fixed polyhedron edge $e$ and the algorithm runs in $O(n^{4+\epsilon})$ time. 

When $s$ moves on $e$, $T_s$ also changes. However, $\Sigma_s$ will not change as long as $T_s$ does not change topologically. We say that $T_s$ {\em changes topologically} if one of the following two types of {\em events} happen: (1) two high-degree vertices of $T_s$ merge into one; (2) a high-degree vertex crosses an edge of $\calP$. Due to our general position assumption, the high-degree vertices in both types of events are degree-3 vertices. 

To compute $\Sigma(e)$, we will first give an algorithm to dynamically maintain $T_s$ as $s$ moves from one end of $e$ to the other. The algorithm runs in  $O((k_1+k_2)\cdot \log n)$ time, where $k_1$ and $k_2$ are the numbers of first type and second type events, respectively. Then, we will prove that both $k_1$ and $k_2$ are bounded by $O(n^{4+\epsilon})$. This will result in an $O(n^{4+\epsilon})$ time algorithm to compute $\Sigma(e)$ (note that the above $\log n$ factor is absorbed by $n^{\epsilon}$). 

Proving $k_1=O(n^{4+\epsilon})$ follows almost directly from the analysis in the grouping method in Section~\ref{sec:group}. However, 
establishing the bound $k_2=O(n^{4+\epsilon})$ proves to be the most significant challenge of this paper. 
We propose the following approach. For each edge $e'\in \calP$ with $e'\neq e$, we wish to bound the number of the second type events on $e'$ (i.e., a degree-3 vertex of $T_s$ crosses $e'$) by $O(n^{3+\epsilon})$, which would lead to $k_2=O(n^{4+\epsilon})$. To this end, we had tried to use the traditional methods such as the source unfolding and the star unfolding, but the best bound we were able to obtain is $O(n^{4+\epsilon})$. We instead propose a new unfolding method. Roughly speaking, we unfold all faces intersecting all possible shortest paths $\pi(s,t)$ with $s\in e$ and $t\in e'$ to a plane containing $e'$. The new unfolding is inspired by the ``one angle one split'' property of \cite{ref:ChenSh90}. Most importantly, like the star unfolding, we discover that there is no ``short-cut'' in the unfolding, i.e., for any image $s'$ of $s\in e$ and any point $t\in e'$, $\Vert s't\Vert\geq d(s,t)$ always holds. Using this property, we prove that the number of second type events on $e'$ is $O(n^{3+\epsilon})$. 

In the rest of this section, we present our algorithm for maintaining $T_s$ in Section~\ref{sec:maintain}. We prove the bounds for $k_1$ and $k_2$ in Sections~\ref{sec:boundk1} and \ref{sec:boundk2}, respectively. 

\subsection{Algorithm for maintaining $T_s$}
\label{sec:maintain}

We now describe our algorithm to dynamically maintain $T_s$ when $s$ moves on $e$ from one end to the other. 

We say that a point $p\in e$ is an {\em event point} if an event happens on $T_s$ when $s$ is at $p$. Correspondingly, there are two types of event points on $e$ corresponding to the two types of events defined above. To maintain $T_s$, we need to compute the event points on $e$. Initially, we compute $T_s$ in $O(n^2)$ time when $s$ is at one end of $e$~\cite{ref:ChenSh90}. 
As $s$ moves on $e$, we will maintain certain information. Note that the movement of each vertex of $T_s$ is an algebraic function of the position of $s\in e$. The function is of constant degree. In other words, as $s$ moves on $e$, the location of each vertex $v$ of $T_s$ in the plane of the face that contains $v$ defines an algebraic curve of constant degree. At each event point, new shortest path edge sequences of $\Sigma(e)$ may also be generated and we need to find them. 

In what follows, we first discuss how to compute the event points on $e$, then explain how to generate shortest path edge sequences of $\Sigma(e)$ using the ridge tree $T_s$, and finally describe the main loop of the algorithm. 

\subsubsection{Computing event points}
To compute the first type event points, we main the following information. We assume that $s$ is not an event point. By our general position assumption, every high-degree vertex of $T_s$ is a degree-3 vertex. For each edge $\xi$ of $T_s$  whose both endpoints are degree-3 vertices, we calculate the next position $p$ of $s$ on $e$ (if any) where $\xi$ will contract to a single point, and we consider $p$ as a {\em candidate} event point and add it to an event point queue $H_s$. Note that such a position $p$ (if it exists) can be calculated in $O(1)$ time. Indeed, let $v_1$ and $v_2$ be two endpoints of $\xi$, and let $f$ be the face of $\calP$ where $\xi$ lies. As the location of $v_1$ (resp., $v_2$) lies on a curve in the plane of $f$ when $s$ moves on $e$ and the curve is of constant size and constant degree, the two curves defined by $v_1$ and $v_2$ intersect $O(1)$ times. Each intersection corresponds to a location on the supporting line of $e$ such that $v_1$ and $v_2$ coincide. Among these $O(1)$ locations, we choose the one closest to $s$ as the next position $p$ for $s$ when $v_1$ and $v_2$ will merge. As such, computing $p$ can be done in $O(1)$ time. 

To compute the second type event points, for each degree-3 vertex $v$ of $T_s$, if $v$ connects to a vertex $v'$ on an edge $e''$ of $\calP$, then we calculate the next position $p$ for $s$ when $v$ will merge with $v'$, i.e., $v$ will be on $e'$. Similarly to the above, this can be done in $O(1)$ time. We also consider $p$ a candidate event point and add $p$ to the event point queue $H_s$.

\subsubsection{Generating shortest path edge sequences from $\boldsymbol{T_s}$}
\label{sec:generate}
We now discuss how to produce shortest path edge sequences of $\Sigma(e)$ by using the ridge tree $T_s$. Our algorithm stores the edge sequences of $\Sigma(e)$ implicitly in a {\em sequence tree $\Psi(e)$}; a similar concept was also discussed in \cite{ref:AgarwalSt97}. Each node $v$ of $\Psi(e)$ stores a polyhedron edge, denoted by $e_v$. In particular, $e_v=e$ if $v$ is the root. The sequence of edges stored in the nodes of the path from the root to $v$ is an edge sequence in $\Sigma(e)$, and let $\sigma_v$ denote that sequence. By definition, for two nodes $u$ and $v$ of $\Psi(e)$, $\sigma_v$ is a prefix of $\sigma_u$ if $v$ is an ancestor of $u$. 
For any point $s\in e$, we let $\Psi_s$ denote the sequence tree storing the edge sequences of $\Sigma_s$. Since $\Sigma_s\subseteq \Sigma(e)$, $\Psi_s$ is a subtree of $\Psi(e)$ with the same root.

The following lemma implies that the edge sequences of $\Sigma_s$ are implicitly determined by $T_s$. 

\begin{lemma}\label{lem:seqtree}
\begin{enumerate}
    \item For each degree-3 vertex $v$ of $T_s$, for any shortest path $\pi(s,v)$, every prefix of the edge sequence of $\pi(s,v)$ is a shortest path edge sequence in $\Sigma_s$. 
    \item For every edge sequence $\sigma$ of $\Sigma_s$, it must be a prefix of the edge sequence of a shortest path $\pi(s,v)$ for some degree-3 vertex $v$ of $T_s$. 
\end{enumerate}    
\end{lemma}
\begin{proof}
The first statement of the lemma is straightforward since $\pi(s,v)$ is a shortest path. In the following, we prove the second statement. 


Let $\overline{T_s}$ denote the union of $T_s$ and shortest paths from $s$ to all polyhedron vertices. We will use the following property of $\overline{T_s}$~\cite{ref:SharirOn86}: For each face $f\in \calP$, for each region $R_f$ of $f$ subdivided by $\overline{T_s}$, there is an edge sequence $\sigma$ such that
every point $t\in R_f$ has a shortest path $\pi(s,t)$ whose edge sequence is $\sigma$. Note that since shortest paths from $s$ to polyhedron vertices do not contain any ridge point, for each edge $\xi$ of $T_s$ that lies on $f$, $\xi$ must be a common edge of two such regions $R_f$. 

Consider a shortest path $\pi(s,t)$ for a point $t\in \calP$. Let $\sigma_t$ denote the edge sequence of $\pi(s,t)$. We can always extend $\pi(s,t)$ to another shortest path $\pi(s,t')$ such that $\pi(s,t)\subseteq \pi(s,t')$ and $t'\in T_s$~\cite{ref:SharirOn86}. As such, $\sigma_t$ is a prefix of $\sigma_{t'}$, the edge sequence of $\pi(s,t')$. Let $\xi$ be the edge of $T_s$ containing $t'$ and let $f_{\xi}$ be the face of $\calP$ where $\xi$ lies. Depending on whether $\xi$ has an endpoint that is a degree-3 vertex of $T_s$, there are two cases. 

\begin{figure}[t]
\begin{minipage}[t]{\linewidth}
\begin{center}
\includegraphics[totalheight=1.7in]{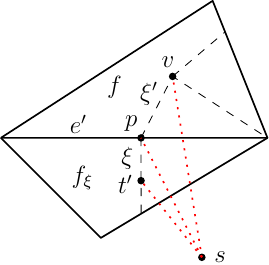}
\caption{
Illustrating the proof of Lemma~\ref{lem:seqtree}: the dashed segments are edges of $T_s$ and the red dotted segments represent shortest paths from $s$ to $t'$, $p$, and $v$, respectively.}
\label{fig:edgeseq}
\end{center}
\end{minipage}
\end{figure}

\begin{itemize}
    \item 
If an endpoint of $\xi$ is a degree-3 vertex $v$, then since $v$ and $t'$ lie on the same edge of $T_s$, they lie in the same region of $f_{\xi}$ subdivided by $\overline{T_s}$. According to the above property, $v$ has a shortest path $\pi(s,v)$ whose edge sequence is the same as $\sigma_{t'}$. In this case, we are done with proving the second lemma statement since $\sigma_t$ is a prefix of $\sigma_{t'}$. 

\item 
If neither endpoint of $\xi$ is a degree-3 vertex, then the two endpoints of $\xi$ lie on two distinct polyhedron edges of $f$, one of which must be the last edge of $\sigma_{t'}$; let $e'$ be the one of the two edges of $f$ that is not the last edge of $\sigma_{t'}$ (see Figure~\ref{fig:edgeseq}). We let $p$ be the endpoint of $\xi$ on $e'$. Then, there is a shortest path $\pi(s,p)$ whose edge sequence is $\sigma_p=\sigma_{t'} \circ \{e'\}$, i.e., concatenate $\sigma_{t'}$ and $\{e'\}$ by appending $e'$ at the end of the $\sigma_{t'}$. Now to prove the second lemma statement, it suffices to show that there exists a degree-3 vertex $v\in T_s$ that has a shortest path edge sequence containing $\sigma_p$ as a prefix. 

The edge $e'$ has two incident polyhedron faces, one of which is $f_{\xi}$, and we let $f$ be the other face. Let $\xi'$ be the edge of $T_s$  in $f$  incident to $p$. Then, $p$ is an endpoint of $\xi'$. If the other endpoint of $\xi'$ is a degree-3 vertex $v$ (see Figure~\ref{fig:edgeseq}), then following above the argument, $v$ has a shortest path from $s$ whose edge sequence is exactly $\sigma_p$ and we are done with proving the second lemma statement. Otherwise, continuing the above argument, we can find another point $p'\in \xi'$ that has a longer shortest path edge sequence $\sigma_{p'}=\sigma_{p}\circ \{e''\}$ for another edge $e''\in \calP$. This process must eventually reach a degree-3 vertex. Indeed, every iteration makes the edge sequence with one  edge longer than before. Since every shortest path edge sequence contains each polyhedron edge at most once, we must eventually reach a degree-3 vertex. This proves the second lemma statement.   
\end{itemize}

Therefore, the second lemma statement is proved. 
\end{proof}

In light of Lemma~\ref{lem:seqtree}, to compute $\Sigma_s$, it suffices to compute edge sequences of all three shortest paths $\pi(s,v)$ for all degree-3 vertices $v$ of $T_s$ (note that each degree-3 vertex $v$ has three shortest paths $\pi(s,v)$). Furthermore, as $s$ moves on $e$, as long as $T_s$ does not change combinatorially, $\Sigma_s$ does not change. Therefore, it suffices to find new shortest path edge sequences when $s$ passes an event point on $e$. To compute $\Psi(e)$, we use the following strategy. Initially when $s$ is at one end of $e$, we compute $\Psi_s$ and set $\Psi(e)=\Psi_s$. Then, we move $s$ on $e$. At each event point, if we find new shortest path edge sequences, we will update $\Psi(e)$ by adding new nodes. After $s$ reaches the other end, $\Psi(e)$ will be computed completely. In what follows, we first discuss how to construct $\Psi_s$ initially when $s$ is at an endpoint of $e$. 

\paragraph{\bf Constructing $\boldsymbol{\Psi_s}$.}
We start with computing the ridge tree $T_s$, which takes $O(n^2)$ time~\cite{ref:ChenSh00}. Recall that $T_s$ has $O(n)$ degree-3 vertices. Since there are three shortest paths from $s$ to each degree-3 vertex, there are $O(n)$ edge sequences in $\Sigma_s$. Since each shortest path edge sequence has $O(n)$ edges, the total size of all edge sequences of $\Sigma_s$ is $O(n^2)$. We can easily construct $\Psi_s$ in $O(n^2)$ time, as follows.  

Suppose we already have the shortest paths from $s$ to all degree-3 vertices of $T_s$. For each path, we add its edge sequence to $\Psi_s$. Initially, $\Psi_s$ consists of a single node storing the edge $e$. Now consider adding the edge sequence of the next path $\pi$. Suppose the edge sequence of $\pi$ is $e=e_1,e_2,\ldots,e_m$. Starting from the root of $\Psi_s$ and the first edge of the sequence, suppose we are considering a node $v\in \Psi_s$ and $e_i$. Recall that $e_v$ denote the edge stored at $v$. 
We assume that $e_v=e_i$, which is true in the beginning of the process for adding $\pi$. If $i=m$, then we are done with adding $\pi$. Otherwise, if $v$ has a child $u$ with $e_u=e_{i+1}$, then proceed with $v=u$ and $e_{i+1}$. If $v$ does not have such a node, then we create a new node $u$ storing $e_{i+1}$ and make $u$ a child of $v$. If $i+1=m$, then we are done with $\pi$; otherwise we  proceed with $v=u$ and $e_{i+1}$. Note that each node $v$ of $\Psi_s$ has at most four children because $e_v$ is incident to two polyhedron faces and each of them has two polyhedron edges other than $e_v$. Therefore, adding a path $\pi$ takes $O(n)$ time, and the total time for constructing $\Psi_s$ is $O(n^2)$. 

As discussed above, to compute $\Sigma(e)$, we initially let $\Sigma(e)=\Sigma_s$ and move $s$ on $e$. If $s$ encounters an event point and a new shortest path edge sequence is identified, we will update $\Psi(e)$ by adding new nodes. Note that we never delete nodes from $\Psi(e)$. To help update $\Psi(e)$, for each degree-3 vertex $v$ of $T_s$, for each shortest path $\pi(s,v)$, we associate $v$ with the node $u\in \Psi(e)$ such that $\sigma_u$ is the edge sequence of $\pi(s,v)$. As there three shortest paths from $s$ to $v$, $v$ is associated with three nodes of $\Psi(e)$. We will discuss how to update $\Sigma(e)$ in the  main loop of the algorithm. 

\subsubsection{Main loop of the algorithm}
\label{sec:mainloop}
The main loop of the algorithm works as follows. In each iteration, we extract the first candidate event point $s_p$ from $H_s$ and process it as follows. 
We first report $s_p$ as a ``true'' event point. Then, depending on whether $s_p$ is the first type event point or the second type, there are two cases. For each case, there are two tasks: update $H_s$ and update $\Psi(e)$. 

\paragraph{\bf The first type event points.}
In the first case, the event point $s_p$ is generated because two degree-3 vertices connecting an edge $\xi$ of $T_s$ merge into one degree-4 vertex at $s=s_p$. After $s$ passes $s_p$, two new degree-3 vertices along with a new edge connecting them will emerge in $T_s$; see Figure~\ref{fig:firsttypeevent}.

\begin{figure}[t]
\begin{minipage}[t]{\linewidth}
\begin{center}
\includegraphics[totalheight=1.3in]{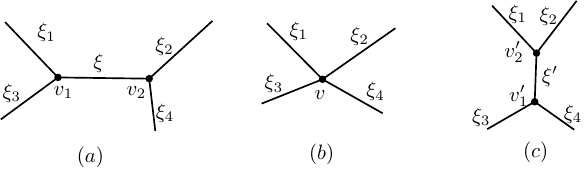}
\caption{
Illustrating the topological changes of $T_s$ at a first type event point $s_p$. 
(a) Before $s$ crosses $s_p$: The two degree-3 vertices $v_1$ and $v_2$ connected by the edge $\xi$ will be merged. (b) $s$ is at $s_p$: $v_1$ and $v_2$ are merged into a degree-4 vertex $v$, and the edge $\xi$ disappears. (c) After $s$ passes $s_p$: a new edge $\xi'$ along with two new degree-3 vertices $v_1'$ and $v_2'$ are created.}
\label{fig:firsttypeevent}
\end{center}
\end{minipage}
\end{figure}

We update $T_s$ by removing $\xi$ and adding the new created edge $\xi'$ and the two new vertices. This is only a local change to $T_s$ and costs $O(1)$ time since the degree of each vertex of $T_s$ is at most $4$ by our general position assumption. After the update of $T_s$, for each edge of $T_s$ that has one of its incident degree-3 vertices updated (including the new edge), we update its corresponding candidate event points in $H_s$ (this includes removing the corresponding old candidate event points from $H_s$ and insert new candidate event points to $H_s$); note that the number of such candidate event points is $O(1)$ because only $O(1)$ edges of $T_s$ are changed. This involves $O(1)$ deletions and insertions for $H_s$. In addition, each incident vertex of $\xi$ is associated with three nodes of $\Psi_s$; we remove these associations. For each incident vertex of $\xi'$, we associate it with the three nodes of $\Psi_s$ whose edge sequences correspond to its three shortest paths from $s$; note that each of these nodes is one of the nodes associated with the two incident vertices of $\xi$ and thus can be found in $O(1)$ time. 

\paragraph{\bf The second type event points.}
In the second case, $s_p$ is generated because a degree-3 vertex $v$ of $T_s$ moves on a polyhedron edge $e'$ at $s=s_p$ and $v$ connects to a vertex $v'$ of $T_s$ on $e'$ by an edge $\xi$ of $T_s$; see Figure~\ref{fig:secondtypeevent}. Let $f$ be the polyhedron face containing $v$  before $s$ crosses $s_p$, and $f'$ the other polyhedron face incident to $e'$.

\begin{figure}[t]
\begin{minipage}[t]{\linewidth}
\begin{center}
\includegraphics[totalheight=1.7in]{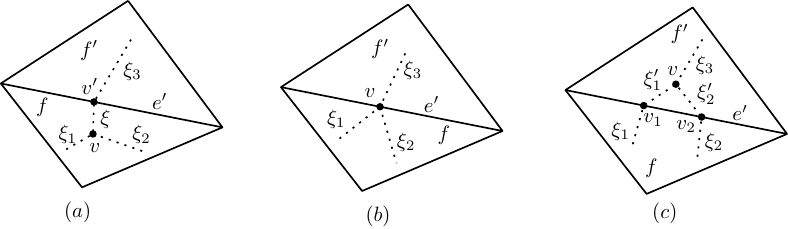}
\caption{
Illustrating the changes of $T_s$ at a second type event point $s_p$. 
(a) Before $s$ crosses $s_p$: a degree-3 vertex $v$ connects a vertex $v'$ of $T_s$ by an edge $\xi$ and $v'$ is on a polyhedron edge $e'$. (b) $s$ is at $s_p$: $v$ is on $e'$ and the edge $\xi$ disappears. (c) After $s$ passes $s_p$: two new edges $\xi_1'$ and $\xi_2'$ of $T_s$ along with two new vertices $v_1$ and $v_2$ on $e'$ are created.}
\label{fig:secondtypeevent}
\end{center}
\end{minipage}
\end{figure}

We first claim that at least one of the three incident edges of $v$ does not intersect $e'$. To see this, one may consider the three edges as portions of the Voronoi edges of three source images of $s$ in an unfolding to the plane of $f$~\cite{ref:MountOn85}. Since each Voronoi region is convex, it is not possible that the three Voronoi edges incident to the Voronoi vertex $v$ all intersect the same line not containing $v$. Therefore, $v$ has at most two incident edges intersecting $e'$. Depending on whether $v$ has one or two incident edges intersecting $e'$, there are two cases (in fact, the two cases are ``inverse'' to each other). 

\begin{itemize}
\item 
If $v$ has one incident edge intersecting $e'$, i.e., at $v'$ (see Figure~\ref{fig:secondtypeevent}), then after $s$ crosses $s_p$, 
$v$ crosses $e'$ and moves to $f'$ (note that if $v$ does not crosses $e'$ and moves back to $f$, then we do not need to process this event because it does not generate any new shortest path edge sequence and we can proceed with the next event point). Recall that $v$ is associated with three nodes of $\Psi_s$ before the event. Two of them are for the two shortest paths from $s$ to $v$ crossing the edge $e'$ to the left and right of $v'$, respectively; let $a$ and $b$ denote these two nodes, and $c$ the third node. By definition, the last edge in the edge sequence of $\sigma_a$ (resp., $\sigma_b$) is $e'$. Let $\pi_a$ (resp., $\pi_b$) be the shortest path from $s$ to $v$ whose edge sequence is $\sigma_a$ (resp., $\sigma_b$). 
After $v$ crosses $e'$ and enters $f'$, $e'$ is not the last edge of the new path $\pi_a$ (resp., $\pi_b$) anymore. The edge sequence of the new path $\pi_a$ is $\sigma_{a'}$, where $a'$ is the parent of $a$ in $\Psi(e)$. Similarly, the edge sequence of the new path $\pi_b$ is $\sigma_{b'}$, where $b'$ is the parent of $b$ in $\Psi(e)$. We remove the associations of $v$ with $a$ and $b$, and instead associate $v$ with $a'$ and $b'$. For $c$, let $\pi_c$ be the shortest path from $s$ to $v$ whose edge sequence is $\sigma_c$. After $v$ crosses $e'$ and enters $f'$, $e'$ becomes the last edge of the edge sequence of $\pi_c$. Hence, we possibly obtain a new edge sequence that is the concatenation of $\sigma_c$ and $e'$. Correspondingly, we add a new child $c'$ storing $e'$ and set it as a new child of $c$. However, before doing so, we check whether $c$ already has a node storing $e'$. If so, then the concatenation of $\sigma_c$ and $e'$ is already stored in $\Psi(e)$ and we do not need to add a new child for $c$. In either case, let $c'$ denote the child of $c$ storing $e'$. We remove the association of $v$ with $c$ and instead assocaite $v$ with $c'$. 

In addition, since $v$ is now in a new face $f'$, we compute new candidate event points associated with $v$, which can be done in $O(1)$ time as $v$ is incident to three edges of $T_s$. We insert them to $H_s$. We also remove the old candidate event points associated with $v$ when $v$ was in $f$. Since two new vertices on $e'$ are created after after $v$ moves to $f'$ (e.g., $v_1$ and $v_2$ in Figure~\ref{fig:secondtypeevent}(c)), we also compute candidate event points associated with them and insert them to $H_s$. 

\item 
If $v$ has two incident edges intersecting $e'$, this can be considered an ``inverse'' case of the above (e.g., in  Figure~\ref{fig:secondtypeevent}, one may consider (c) as the situation before $s$ crosses $s_p$ and consider (a) as the situation after $s$ crosses $s_p$). Our algorithm for this case basically reverses the operations in the above case. Specifically, let $c$ be the associated node of $v$ in $\Psi_s$ whose edge sequence $\sigma_c$ is the one for the shortest path from $s$ to $v$ through the portion of $e'$ between the two intersections of $e'$ with the two incident edges of $v$. Let $a$ and $b$ be the other two associated nodes of $v$, respectively. After $v$ crosses $e'$, we remove the association of $v$ with $c$ and instead associate $v$ with $c$'s parent. For $a$, if $a$ does not have a child storing $e'$, then we add such a child. In either case, let $a'$ be the child of $a$ storing $e'$. We remove the association of $v$ with $a$ and instead associate $v$ with $a'$. We do the same for $b$. Finally, we update the associated candidate event points for $v$ and the two intersections between $e'$ and the two incident edges of $v$. 
\end{itemize}

The algorithm stops once $s$ reaches the other end of $e$, after which all event points on $e$ are computed and $\Psi(e)$ is constructed completely. 

For the time analysis, since processing each event involves $O(1)$ candidate event point updates on $H_s$, 
it takes $O(\log |H_s|)$ time to update $H_s$ for each event. Note that $|H_s|$ is always bounded by $O(|T_s|)$ since each vertex of $T_s$ defines $O(1)$ candidate events in $H_s$. As $|T_s|=O(n^2)$, we obtain $|H_s|=O(n^2)$, and thus $\log |H_s|=O(\log n)$. Also, updating $\Psi(e)$ for each event takes $O(1)$ time because only $O(1)$ nodes of $\Psi(e)$ are updated. 
Therefore, processing each event takes $O(\log n)$ time and the total time of the algorithm for computing $\Psi(e)$ is $O((k_1+k_2)\cdot \log n)$.

In Section~\ref{sec:boundk1} and \ref{sec:boundk2}, we will prove that both $k_1$ and $k_2$ are bounded by $O(n^{4+\epsilon})$. This leads to an $O(n^{4+\epsilon})$ time algorithm to compute $\Sigma(e)$. Applying the algorithm to all polyhedron edges $e$ will compute the exact set of all shortest path edge sequences on $\calP$ in $O(n^{5+\epsilon})$ time. We summarize the result in the following theorem. 

\begin{theorem}
Given a convex polyhedral surface $\calP$ of $n$ vertices, the exact set of all shortest path edge sequences on $\calP$ can be computed in $O(n^{5+\epsilon})$ time. In particular, for each edge $e\in \calP$, the exact set of all shortest path edge sequences on $\calP$ whose first edge is $e$ can be computed in $O(n^{4+\epsilon})$ time. 
\end{theorem}

\subsection{Proving $\boldsymbol{k_1=O(n^{4+\epsilon})}$}
\label{sec:boundk1}

If $s$ is at a first type event point, then two degree-3 vertices of $T_s$ are merged into a degree-4 vertex $v$, which has four shortest paths from $s$. This implies that the pair $(s,v)$ corresponds to a vertex in the lower envelope of the functions defined by the $n$ source images of $s$ in the plane of the kernel $\calK_{\eta}$ of the star unfolding $\bigstar_{\eta}$ for the edgelet $\eta$ that contains $s$ on $e$ (recall in Section~\ref{sec:group} that each source image $s_i$ defines a 3-variate algebraic function $f_i(s,t^*)$ for $s\in \eta$ and $t^*$ in the plane of $\calK_{\eta}$). Therefore, $k_1$ is no more than the total number of vertices of all lower envelopes for all edgelets on $e$. Since $e$ has $O(n^2)$ edgelets~\cite{ref:AgarwalSt97}, following the grouping technique in Section~\ref{sec:group}, we can partition the edgelets of $e$ into $O(n)$ groups such that each group defines $O(n)$ functions. Hence, the number of vertices in the lower envelope of all functions in each group is $O(n^{3+\epsilon})$~\cite{ref:SharirAl94}. Therefore, the total number of vertices in the lower envelopes in all $O(n)$ groups is $O(n^{4+\epsilon})$. This proves $k_1=O(n^{4+\epsilon})$.

\subsection{Proving $\boldsymbol{k_2=O(n^{4+\epsilon})}$}
\label{sec:boundk2}

We now prove the bound for $k_2$, which is significantly more challenging. By definition, each second type event point corresponds to 
a pair $(s,t)$ such that $s\in e$ and $t\in e'$ for another polyhedron edge $e'$ with $e\neq e'$ and there are three shortest \st\ paths; we refer to the event point $s$ an {\em $e'$-event}. 
Consider an edge $e'$ of $\calP$ with $e'\neq e$. Let $k_{e'}$ denote the number of $e'$-events. In the following, we will prove $k_{e'}=O(n^{3+\epsilon})$, which leads to the bound $k_2=O(n^{4+\epsilon})$. For convenience, we consider $e'$ an open segment without including its endpoints. 


Roughly speaking, our strategy is to find the edge sequences of all possible shortest paths $\pi(s,t)$ with $s\in e$ and $t\in e'$ and unfold these edge sequences into a single plane $\Pi_{e'}$ containing $e'$. The plane $\Pi_{e'}$ contains a number of source images of $s\in e$ so that the location of each source image is a linear function of the location of the source $s\in e$. Each source image $s'$ defines an algebraic function that is the distance from $s'$ to any point on $e'$ in $\Pi_{e'}$, and the function is bivariate and is of constant degree and constant size. We will show that bounding $k_{e'}$ boils down to bounding the number of vertices in the lower envelope of these functions. For a set of $n$ such functions, the number of vertices in their lower envelope is bounded by $O(n^{2+\epsilon})$~\cite{ref:HalperinNe94,ref:SharirAl94}. Using the grouping technique as in Section~\ref{sec:group} or in Section~\ref{sec:boundk1}, we finally show that $k_{e'}=O(n^{3+\epsilon})$. The details are given below.

Let $f_1$ and $f_2$ be the two incident polyhedron faces of $e'$, respectively. We say a path $\pi$ from $s\in e$ to $t\in e'$ is a {\em $f_1$-constrained} path if the path does not intersect the interior of the face $f_2$ (it was called {\em $f_2$-free path} in \cite{ref:MitchellTh87}). We define the {\em $f_1$-constrained} shortest \st\ path as a shortest path among all $f_1$-constrained \st\ paths, and we use $\pi_{f_1}(s,t)$ to denote it. We define an {\em $f_2$-constrained shortest path} $\pi_{f_2}(s,t)$ analogously. Clearly, the shorter one of $\pi_{f_1}(s,t)$ and $\pi_{f_2}(s,t)$ is a shortest \st\ path~\cite{ref:MitchellTh87}. 

\paragraph{Admissible interval sets.}
For a fixed point $s\in e$, we wish to compute a set $\calI_s$ of intervals (i.e., segments) on $e'$ such that 
(1) each interval $I\in \calI_s$ is associated with an edge sequence $\sigma_I$ such that there is a geodesic (not necessarily shortest) $f_1$-constrained path $\pi_I(s,t)$ from $s$ to $t$ whose edge sequence is exactly $\sigma_I$ for any point $t\in I$; (2) for any point $t\in e'$, if $\pi_{f_1}(s,t)$ is a shortest \st\ path, then there is an interval $I\in\calI_s$ such that $t\in I$ and $\pi_I(s,t)$ is $\pi_{f_1}(s,t)$. We call $\calI_s$ an {\em $f_1$-admissible interval set} of $e'$. 

Now suppose that $s$ moves on $e$. We wish to understand combinatorial changes of $\calI_s$ during the movement of $s$. 
We say that two $f_1$-admissible interval sets $\calI_s$ and $\calI_{s'}$ are {\em combinatorially equivalent} if $|\calI_s|=|\calI_{s'}|$ and each interval $I$ of $\calI_s$ corresponds to an interval $I'\in \calI'$ such that $\sigma_I=\sigma_{I'}$ (note that the lengths of $I$ and $I'$ may be different; we also say these two intervals are {\em combinatorially equivalent}). 

A segment of $e$ is called an {\em elementary segment} if $\calI_s$ is combinatorially equivalent for all points $s$ in the segment. We will show that it is possible to partition $e$ into $O(n^2)$ elementary segments. Using the elementary segments, we will prove the bound $k_{e'}=O(n^{3+\epsilon})$.

In the following, for any fixed point $s\in e$, we first adapt Chen and Han's algorithm~\cite{ref:ChenSh90} (referred to as the CH algorithm) in Section~\ref{sec:ch} to compute an $f_1$-admissible interval set $\calI_s$ of size $O(n)$ in $O(n^2)$ time. Then, we show in Section~\ref{sec:eleint} that $e$ can be partitioned into $O(n^2)$ elementary segments. We finally prove the bound $k_{e'}=O(n^{3+\epsilon})$ in Section~\ref{sec:bound}. 

\subsubsection{Computing an $\boldsymbol{f_1}$-admissible set $\boldsymbol{\calI_s}$}
\label{sec:ch}
Let $n_f$ denote the number of faces of $\calP$. 
The CH algorithm is based on the following two properties of shortest paths: (1) any shortest path on $\calP$ intersects at most $n_f$ faces; (2) the interiors of any two shortest paths from $s$ do not intersect. 

The CH algorithm builds a tree $\Gamma_s$ to store edges sequences. The tree (which is called ``sequence tree'' in \cite{ref:ChenSh90}) is similar to but not the same as the sequence tree defined in Section~\ref{sec:generate}; for differentiation, we call it the {\em CH-tree}. Its root is $s$ itself and stores the edge $e$. Every other node $v$ of $\Gamma_s$ stores an edge $e_v$ such that (1) $e_v$ and $e_{p_v}$ are in the same face of $\calP$, where $p_v$ is the parent of $v$, and (2) there is a geodesic path $\pi_v(s,t)$ from $s$ to a point $t\in e_v$ and the edge sequence of the path is exactly the sequence of edges stored in the nodes of the path of $\Gamma_s$ from the root to $e_v$ (we use $\sigma_v$ to denote the edge sequence). In other words, if we unfold the faces intersecting the geodesic path $\pi_v(s,t)$ to the plane $\Pi_v$ of the face containing $e_v$ and $e_{p_v}$, then there is a line segment connecting the image of $s$ to $t$ inside the union $P_v$ of the images of the above faces in the unfolding (i.e., the image of $s$ is visible to $t$ in $P_v$); note that $P_v$ is a simple polygon. We also store at $v$ the image of $s$ in the unfolding plane $\Pi_v$, denoted by $s_v$. In addition, since $P_v$ is a simple polygon, all points of $e_v$ visible to $s$ form an interval, denoted by $I_v$ and called the {\em projection interval} of $s$ on $e_v$, which is also associated with $v$; see Figure~\ref{fig:triple}. As such, each node $v\in T_s$ is associated with a triple $(e_v,s_v,I_v)$. 
If $v$ is the root, then the triple is $(e,s,e)$. The algorithm works as follows.

\begin{figure}[t]
\begin{minipage}[t]{\linewidth}
\begin{center}
\includegraphics[totalheight=1.7in]{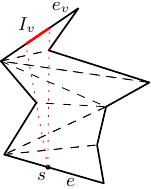}
\caption{
The polygon with solid boundary is $P_v$. The edge sequence $\sigma_v$ consists of the dashed segments as well as $e$ and $e_v$. The red segment on $e_v$ is $I_v$.}
\label{fig:triple}
\end{center}
\end{minipage}
\end{figure}

Initially, we set $s$ as the root of $\Gamma_s$ and associate the triple $(e,s,e)$ with it. The algorithm runs in $n_f$ iterations. Each $i$-th iteration will construct nodes in the $(i+1)$-th level (the root is in the first level). In the $i$-th iteration, for each leaf $v=(e_v,s_v,I_v)$ in the $i$-th level of the tree, we do the following. Note that $e_v$ is incident to two faces, one containing both $e_v$ and $e_{p_v}$, where $p_v$ is the parent of $v$, denoted by $f_v'$, and the other face is called the {\em shadowed face} of $e_v$ and denoted by $f_v$. Let $f_v=\triangle abc$ with $e_v=\overline{bc}$; see Figure~\ref{fig:projection}. We first unfold $s_v$ about $e_v$ to the plane of $f_v$ (so that the image of $s_v$ and $f_v$ are in the opposite side of $e_v$) and let $s_v'$ be the image of $s_v$ in the unfolding. For the edge $\overline{ab}$, we create a child $u$ for $v$ storing $\overline{ab}$. Then, based on $I_v$ and $s_v'$, we can determine the projection interval $I_u$ on $\overline{ab}$. We associate the triple $(\overline{ab},s_v',I_u)$ with $u$. Note that if $I_u$ is empty, then we do not create $u$. We do the same for the edge $\overline{ac}$.

\begin{figure}[t]
\begin{minipage}[t]{\linewidth}
\begin{center}
\includegraphics[totalheight=1.6in]{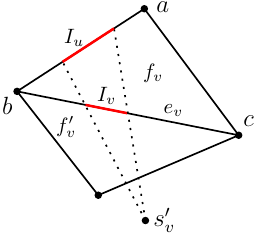}
\caption{
Illustrating the computation of a node $u$ associated with triple $(\overline{ab},s_v',I_u)$.}
\label{fig:projection}
\end{center}
\end{minipage}
\end{figure}

After $n_f$ iterations, for each node $v$ with $e_v=e'$ and $f_v'=f_1$, we add the projection interval $I_v$ to $\calI_s$. We can prove that $\calI_s$ is indeed an $f_1$-admissible interval set for $e'$ (see Lemma~\ref{lem:admissible}). However, the tree $\Gamma_s$ computed in the above algorithm may have an exponential number of nodes because each leaf of the tree in each iteration creates two children. We can reduce the size by the following observation (called ``one angle one split'' property in \cite{ref:ChenSh90}). Suppose $v$ and $u$ are two nodes with $e_v=e_u=\overline{bc}$ and $f_v=f_u=\triangle abc$. Then, at most one of them can have two children. Indeed, if $a$ is not in the projection interval of $s_v'$ or $a$ is not in the projection interval of $s_u'$, then it is vacuously true that only one of $u$ and $v$ can have two children. Otherwise, assume that $\overline{bc}$ is horizontal, $b$ is left of $c$, and $s_v'$ is left of $s_u'$; see Figure~\ref{fig:anglesplit}. 
If $\Vert s_v'a\Vert\leq \Vert s_u'a \Vert$, then the key observation is that any point on $\overline{ab}$ has a shorter path from $s$ through $s_v'$ than through $s_u'$, and thus, we do not need to create a child for $u$ storing the edge $\overline{ab}$. In this case, we say that $v$ {\em occupies} the polyhedron vertex $a$. Similarly, if $\Vert s_v'a\Vert\geq \Vert s_u'a\Vert$, then we do not need to create a child for $v$ storing $\overline{ac}$. 

\begin{figure}[t]
\begin{minipage}[t]{\linewidth}
\begin{center}
\includegraphics[totalheight=1.2in]{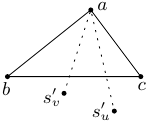}
\caption{
Illustrating the ``one angle one split'' property~\cite{ref:ChenSh90}.}
\label{fig:anglesplit}
\end{center}
\end{minipage}
\end{figure}

With the observation, we can reduce the size of $\Gamma_s$ by modifying the above algorithm as follows. When processing $v$ in the algorithm, if the projection of $s_v'$ does not cover $a$, then we do the same as before, which will result in exactly one child of $v$. Otherwise, if $a$ was not occupied by any other node, then we do the same as before, which will create two children of $v$. If $a$ was occupied by another node $u$, then we compare $\Vert s_v'a\Vert$ and $\Vert s_u'a\Vert$. If $\Vert s_v'a\Vert\geq \Vert s_u'a\Vert$, then we only create one child for $v$, storing the edge $\overline{ab}$. Otherwise, we create two children for $v$ and delete the subtree at the child of $u$ storing $\overline{ab}$.

In this way, the total number of nodes of the resulting tree $\Gamma_s$ is $O(n^2)$ because there are $O(n)$ leaves at each iteration due to the ``one angle one split'' property. As before, for each node $v$ with $e_v=e'$ and $f_v'=f_1$, we add the projection interval $I_v$ to $\calI_s$. 
Furthermore, for each interval $I_v\in \calI_s$, if $v$ has an ancestor $u$ with $e_u=e'$, then the edge sequence $\sigma_v$ contains the polyhedron edge $e_v$ more than once, meaning that $\sigma_v$ cannot be a shortest path edge sequence; in this case, we remove $I_v$ from $\calI_s$ and we call this the {\em post-pruning} strategy. Due to this strategy, we have the following observation.

\begin{observation}\label{obser:postpruning}
    For each leaf $u$ of $\Gamma_s$, there is at most one node $v$ in the path of $\Gamma_s$ from the root to $u$ such that the projection interval $I_v$ is in $\calI_s$. 
\end{observation}
\begin{proof}
Assume to the contrary that there are two nodes $v$ and $v'$ in the path of $\Gamma_s$ from the root to $u$ such that both projection intervals $I_v$ and $I_{v'}$ are in $\calI_s$. Then, one of $v$ and $v'$ must the ancestor of the other. Without loss of generality, we assume that $v$ is the ancestor of $v'$. According to our post-pruning strategy, $I_{v'}$ cannot be in $\calI_s$, a contradiction. 
\end{proof}

Since $\Gamma_s$ has $O(n)$ leaves, by the above observation, we have $|\calI_s|=O(n)$. 




\begin{lemma}\label{lem:admissible}
    $\calI_s$ is an $f_1$-admissible interval set for $e'$. 
\end{lemma}
\begin{proof}
To prove the lemma, we need to prove that $\calI_s$ has the two properties in the definition of $f_1$-admissible interval sets. 
For each interval $I\in \calI_s$, for any point $t\in I$, according to our algorithm, there is a geodesic $f_1$-constrained \st\ path whose edge sequence is $\sigma_I$. Therefore, the first property holds. We prove the second property below. 

For any point $t\in e$, suppose that $\pi_{f_1}(s,t)$ is a shortest \st\ path. Let $\sigma$ be its edge sequence. Let $\Gamma_s'$ be the CH-tree constructed in our first algorithm that does not use the ``one angle one split'' property to prune the sequence tree. 
Since the path is a geodesic path intersecting at most $n_f$ polyhedron faces, according to our first algorithm, there must be a node $v$ in $\Gamma_s'$ with $\sigma_v=\sigma$ and $t\in I_v$. In our second algorithm, some branches of $\Gamma_s$ are pruned. However, if a path is pruned, then the path cannot be a shortest path. Since $\pi_{f_1}(s,t)$ is a shortest path, it can never be pruned. Hence, the node $v$ is still in $\Gamma_s$, which implies that $\pi_{I_v}(s,t)$ is $\pi_{f_1}(s,t)$. Finally, since $\pi_{f_1}(s,t)$ is a shortest path, no two edges in $\sigma_v$ are identical, meaning that $I_v$ will ``survive'' the post-pruning strategy. Therefore, the interval $I_v$ must be in $\calI_s$. 
This proves the second property. 
\end{proof}


\subsubsection{Partitioning $\boldsymbol{e}$ into $\boldsymbol{O(n^2)}$ elementary segments}
\label{sec:eleint}
We now show that $e$ can be partitioned into $O(n^2)$ elementary segments. To this end, we consider the following question: When will $\calI_s$  change topologically if $s$ moves on $e$? 
First of all, an interval $I_v$ of $\calI_s$ for a node $v\in \Gamma_s$ may disappear because the visibility of $s_v$ on $e'$ in the unfolding polygon $P_v$ is blocked. An interval may also disappear because the $f$-constrained shortest path from $s$ to a polyhedron vertex $v$ is topologically changed (i.e., the edge sequence is changed) for a polyhedron vertex-face pair $(v,f)$ with $v$ being a vertex of $f$. The latter reason may also cause a new interval to be added to $\calI_s$. The details are discussed below. 


We first partition $e$ into $O(n^2)$ segments (we call them {\em quasi-elementary segments}) so that in each segment $\tau$, for all vertex-face pairs $(v,f)$ of $\calP$, $s$ has topologically the same $f$-constrained shortest paths from $s$ to $v$ for all points $s\in \tau$. This is done in the following lemma. 


\begin{lemma}\label{lem:quasiinterval}
It is possible to partition $e$ into $O(n^2)$ quasi-elementary segments. 
\end{lemma}
\begin{proof}
Consider a polyhedron vertex-face pair $(v,f)$. Let $e_f$ denote the edge of $f$ opposite to $v$. Let $v_1$ and $v_2$ be the other two vertices of $f$; see Figure~\ref{fig:vertexface}. 

\begin{figure}[t]
\begin{minipage}[t]{\linewidth}
\begin{center}
\includegraphics[totalheight=1.0in]{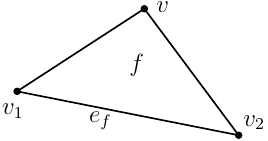}
\caption{
Illustrating a vertex-face pair $(v,f)$.}
``````````\label{fig:vertexface}
\end{center}
\end{minipage}
\end{figure}

We adapt the algorithm in \cite{ref:MitchellTh87} (called MMP algorithm) to compute $f$-constrained shortest paths from $s$ to $v$ for all points $s\in \calP$. To adapt the algorithm, we start the algorithm from the vertex $v$ and change the initial step as follows (using the terminology in \cite{ref:MitchellTh87}): Create a candidate interval whose extent is $\overline{v_1v_2}$ and whose root is $v$ (the original algorithm does this for all edges of $\calP$ opposite to $v$ in all incident faces of $v$, while here we only do this for $\overline{v_1v_2}$). The rest of the algorithm is the same as before. 
An alternative way to adapt the algorithm is the following. Add a rectangular face of infinitely long with an edge coincident with $\overline{vv_1}$ and the face extends out of $\calP$. Do the same for the edge $\overline{vv_2}$. We consider $v$ as a point in the interior of the region bounded by $f$ and the two new faces. As such, these two new faces play the role of ``walls'' to make sure paths from $v$ only go through $f$ initially. Note that the polyhedron with the two walls are not convex anymore, but the MMP algorithm works for non-convex polyhedra. 

The above adapted MMP algorithm subdivides each polyhedron edge (in particular, the edge $e$) into $O(n)$ intervals such that $f$-constrained shortest $v$-$s$ paths for all $s$ in the same interval have the same edge sequence (note that the path may contain either $\overline{vv_1}$ or $\overline{vv_2}$, in which case the vertex $v_1$ or $v_2$ is considered a special ``vertex edge'' of the edge sequence). 

We do the above for all polyhedron vertex-face pairs $(v,f)$. As there are $O(n)$ such pairs, the total number of intervals on $e$ is $O(n^2)$. The endpoints of these intervals partition $e$ into $O(n^2)$ segments such that for each segment $\tau$, for all vertex-face pairs $(v,f)$, $s$ has topologically equivalent $f$-constrained shortest $s$-$v$ paths for all points $s\in \tau$. This proves the lemma. 
\end{proof}

Suppose that we partition $e$ into quasi-elementary segments by Lemma~\ref{lem:quasiinterval}. Consider a quasi-elementary segment $\tau$. Suppose that we have the CH-tree $\Gamma_s$ when $s$ is at one end of $\tau$. Now image that we move $s$ in $\tau$. During the moving, since $\tau$ is a quasi-elementary segment, by definition, for all vertex-face pairs $(v,f)$, $s$ has topologically equivalent $f$-constrained shortest $s$-$v$ paths. According to our algorithm for constructing $\Gamma_s$, $\Gamma_s$ does not change except in the following situations: (1) a projection interval $I_v$ of a node $v\in \Gamma_s$ become empty because $s_v$ is not visible to $e_v$  anymore in the unfolding polygon $P_v$; (2) a projection interval $I_v$ of a node $v\in \Gamma_s$ emerges because $s_v$ becomes visible to $e_v$ in $P_v$. Observe that if $s_v$ is visible to $e_v$ in $P_v$, then $s_u$ must be visible to $e_u$ in $P_u$ for an ancestor $u$ of $v$ in $\Gamma_s$. This implies that in the above first situation, $v$ must be a leaf of $\Gamma_s$, while in the second situation $v$ becomes a new leaf of $\Gamma_s$. Hence, during the moving of $s$ in $\tau$, to update $\Gamma_s$, in the first case we remove $v$ from $\Gamma_s$ while in the second case we add $v$ as a new leaf of $\Gamma_s$. We say that other nodes of $\Gamma_s$ do not change although the lengths of the projection intervals may change. In general, for two points $s$ and $s'$ on $e$, we consider a node $v\in \Gamma_s$ and a node $v'\in \Gamma_{s'}$ the same node if $e_v=e_{v'}$, $\sigma_v=\sigma_{v'}$, and both $I_v$ and $I_{v'}$ are non-empty. We now define $\Gamma_e$ as the union of $\Gamma_s$ for all points $s\in e$. In the following, we show that (1) $\Gamma_s$ for any point $s\in e$ is a subtree of $\Gamma_e$ with the same root, and (2) any leaf of $\Gamma_e$ is a leaf of $\Gamma_s$ for some point $s\in e$. 

Suppose $\tau$ is the first quasi-elementary segment of $e$. Let $s$ be the common endpoint of $\tau$ and $e$. Intially, set $\Gamma_e=\Gamma_s$. Then, we move $s$ on $e$. During the movement, we will update $\Gamma_e$ by only adding new nodes. If $s$ moves in $\tau$, one of the two situations discussed above will happen. In the first situation, we do not change $\Gamma_e$. In the second situation, however, we add a new leaf $v$ as a child of a current node $u$ of $\Gamma_e$. As such, when $s$ reaches the other end of $\tau$, the current $\Gamma_e$ has the above two properties because all updates to $\Gamma_e$ only involve leaves. Furthermore, we have the following observation. 

\begin{observation}
    When $s$ moves in $\tau$, the number of leaves in $\Gamma_e$ does not change. 
\end{observation}
\begin{proof}
According to our analysis, $\Gamma_e$ only changes in the second situation, in which a new leaf $v$ is added to $\Gamma_e$ as a child of a node $u$. We claim that $u$ was a leaf before the update. Assume to the contrary this is not true. Then, before the update, $u$ has a child $v'$. Let $\triangle abc$ be the shadow face of $u$ with $e_u=\overline{bc}$. Then, $e_{v'}$ is either $\overline{ab}$ or $\overline{ac}$. Note that $u$ cannot have two children before the update since otherwise $v$ would not become a child of $u$ after the update (this is further because every node of $\Gamma_s$ can have at most two children). Since $u$ only had one child before the update, $u$ cannot occupy the polyhedron vertex $a$. Since $u$ has two children after the update, $u$ must occupy $a$. This means the $f$-constrained shortest path from $s$ to $a$ with $f=\triangle abc$ has been changed topologically. But this contradicts with the fact that $\tau$ is a quasi-elementary segment (and $s$ moves inside $\tau$). 
\end{proof}

Now consider the situation where $s$ cross the common endpoint $z$ between $\tau$ and the next quasi-elementary segment $\tau'$ on $e$. As $s$ crosses $z$, the $f$-constrained shortest path from $s$ to $a$ changes topologically for some polyhedron vertex-face pair $(a,f)$. Let $e''=\overline{bc}$ denote the edge of $f$ opposite to $a$. Let $V$ be the set of nodes $v\in \Gamma_e$ whose associated edge $e_v$ is $e''$ and whose shadowed face is $f$. Based on the algorithm for constructing $\Gamma_s$, among all nodes of $V$, only one node $v_1$ that occupies $a$ can have two children, which is the node that determines the $(a,f)$-constrained shortest path from $s$ to $a$. After $s$ crosses $z$, another node $v_2\in V$ could occupy $a$ and have two children. As such, one of the subtrees of $v_1$ may be clipped off and a new subtree may be added to $v_2$. Since the $f'$-constrained shortest paths from $s$ to $v'$ for all other vertex-face pairs $(v',f')$ do not change topologically, each node of the new subtree of $v_2$ must have a single child (since if a node has two children, then it must occupy a polyhedron vertex $v'$, meaning that the $f'$-constrained shortest paths from $s$ to $v'$ has changed for some incident face $f'$ of $v'$). Therefore, the new subtree at $v_2$ has exactly one leaf. Hence, after $s$ crosses $z$ and enters $\tau'$, $\Gamma_e$ only gain at most one more leaf. Also, since a new subtree is added to a current node $v_2$ of $\Gamma_e$ after $s$ enters $\tau'$, the above two properties of $\Gamma_e$ still hold. 

According to the above discussion, since there are $O(n^2)$ quasi-elementary segments on $e$, $\Gamma_e$ has at most $O(n^2)$ leaves. 

For any point $s\in e$, recall that each interval $I\in \calI_s$ corresponds to an edge sequence $\sigma_I$ that is the edge sequence of a geodesic path from a point $s\in e$ to a point $t\in e'$. We define $\Lambda_s=\{\sigma_I\ |\ I\in \calI_s\}$ and  $\Lambda_e=\bigcup_{s\in e}\Lambda_s$, i.e., $\Lambda_e$ is the set of distinct edge sequences $\sigma_I$ of $\Lambda_s$  of all points $s\in e$. Since $\Gamma_s$ is a subtree of $\Gamma_e$ with the same root, we have the following lemma. 

\begin{lemma}\label{lem:edgeseq}
\begin{enumerate}
    \item $|\Lambda_e|=O(n^2)$. 
    \item Each edge sequence $\sigma\in \Lambda_e$ corresponds to a segment $\tau_{\sigma}\subseteq e$ such that $\calI_s$ contains an interval $I$ with $\sigma_I=\sigma$ if and only if $s\in \tau_{\sigma}$.
\end{enumerate}
\end{lemma}
\begin{proof}
Consider a leaf $u$ of $\Gamma_e$. According to our algorithm for constructing $\Gamma_e$, $u$ is a leaf of $\Gamma_s$ for some point $s\in e$. By Observation~\ref{obser:postpruning}, among all nodes in the path $\pi_u$ from $u$ to the root of $\Gamma_s$, there is at most one node $v$ whose projection interval $I_v$ is in $\calI_s$, and thus there is at most one node $v$ in $\pi_u$  whose edge sequence $\sigma_v$ is in $\Lambda_e$. 

On the other hand, consider an edge sequence $\sigma\in \Lambda_e$, which must be $\sigma_v$ for some node $v\in \Gamma_s$ for some point $s\in e$. Since $\Gamma_s$ is a subtree of $\Gamma_e$ with the same root, there is a leaf $u\in \Gamma_e$ such that $v$ is a node of $\Gamma_e$ in the path from the root to $u$. 

The above implies that $|\Lambda_e|$ is no more than the number of leaves of $\Gamma_e$, which is $O(n^2)$. This proves the first lemma statement. 

For the second lemma statement, consider an edge sequence $\sigma\in \Lambda_e$, which must be $\sigma_v$ for some node $v\in \Gamma_s$ for some point $s\in e$. Recall that the unfolding polygon $P_v$ is a simple polygon that has $e'$ as an edge, and one of its edges, denoted by $e''$, is the image of $e$ that contains $s_v$. Since $P_v$ is a simple polygon and $e'$ and $e''$ are two edges of $P_v$, $e''$ contains a single segment $\tau''$ that consists of points of $e''$ weakly visible to $e'$ (a point $p\in P_v$ is {\em weakly visible} to $e'$ if $p$ is visible to at least one point on $e'$)~\cite{ref:GuibasLi87}. Note that $\calI_s$ has an interval $I$ with $\sigma_I=\sigma$ if and only if $s_v$ is weakly visible to $e'$. Therefore, $\calI_s$ has an interval $I$ with $\sigma_I=\sigma$ if and only if $s$ is in the preimage of $\tau''$, which is a segment of $e$. This proves the second lemma statement. 
\end{proof}

For ease of exposition, we assume that the endpoints of the segments $\tau_{\sigma}$ for all edge sequences $\sigma\in \Lambda_e$ on $e$ are distinct. 

By the definition of elementary segments, for any elementary segment $\tau$ of $e$, 
$\Lambda_s$ is the same for all points $s\in \tau$, and thus we define $\Lambda_{\tau}$ as $\Lambda_s$ for any point $s\in \tau$. 
Recall that $|\calI_s|=O(n)$. Hence, $|\Lambda_{\tau}|=|\Lambda_{s}|=|\calI_s|=O(n)$. 
Then, we have the following lemma. 

\begin{lemma}\label{lem:eleinterval}
    The edge $e$ can be partitioned into $O(n^2)$ elementary segments such that for every two adjacent segments $\tau$ and $\tau'$, $\Lambda_{\tau}$ and $\Lambda_{\tau'}$ differ by at most one edge sequence. 
\end{lemma}
\begin{proof}
For each edge sequence $\sigma\in \Lambda_e$, we call the segment $\tau_{\sigma}$ in Lemma~\ref{lem:edgeseq} a {\em basic segment}. As $|\Lambda_e|=O(n^2)$, there are $O(n^2)$ basic segments and their endpoints partition $e$ into $O(n^2)$ segments $\tau$ such that every basic segment either completely contains $\tau$ or is disjoint from the interior of $\tau$. By Lemma~\ref{lem:edgeseq}, $\calI_s$ is topologically equivalent for all points $s\in \tau$. Hence, $\tau$ is an elementary segment on $e$. Further, for any two such segments $\tau$ and $\tau'$ that are adjacent, the subset of basic segments containing $\tau$ differ from the subset of basic segments containing $\tau'$ by at most one element. By Lemma~\ref{lem:edgeseq}, this implies that $\Lambda_{\tau}$ and $\Lambda_{\tau'}$ differ by at most one edge sequence. The lemma thus follows. 
\end{proof}

The following lemma summarizes our result in this subsection. 

\begin{lemma}\label{lem:summaryf1}
The edge $e$ can be partitioned into $O(n^2)$ elementary segments with respect to the face $f_1$ such that the following properties hold.
\begin{enumerate}
    \item For each elementary segment $\tau$, $|\Lambda_{\tau}|=O(n)$. 
    \item  For every two adjacent elementary segments $\tau$ and $\tau'$, $\Lambda_{\tau}$ and $\Lambda_{\tau'}$ differ by at most one edge sequence. 
    \item  For any elementary segment $\tau$, for any point $s\in \tau$ and any edge sequence $\sigma\in \Lambda_{\tau}$, there are a point $t\in e'$ and a geodesic \st\ path whose edge sequence is $\sigma$.
    \item For any elementary segment $\tau$, for any point $s\in \tau$ and any point $t\in e'$, if $\pi(s,t)$ is an $f_1$-constrained path, then the edge sequence of $\pi(s,t)$ is in $\Lambda_{\tau}$.
\end{enumerate}
\end{lemma}
\begin{proof}
We have already proved the first lemma statement. The second lemma statement is proved in Lemma~\ref{lem:eleinterval}. For the third lemma statement, consider a point $s\in \tau$ and an edge sequence $\sigma\in \Lambda_{\tau}$. By definition, there is an interval $I\in \calI_s$ whose edge sequence $\sigma_I$ is $\sigma$. By Lemma~\ref{lem:admissible}, $\calI_s$ is an $f_1$-admissible interval set for $e'$. Hence, for any point $t\in I$, there is a geodesic $f_1$-constrained \st\ path whose edge sequence is $\sigma$. Note that a geodesic $f_1$-constrained path is a geodesic path. This proves the third lemma statement. For the fourth one, since $\pi(s,t)$ is an $f_1$-constrained path, i.e., $\pi_{f_1}(s,t)$ is $\pi(s,t)$, and $\calI_s$ is an $f_1$-admissible interval set for $e'$, by the definition of $f_1$-admissible interval sets, there is an interval $I\in \calI_s$ such that $t\in I$ and $\pi_I(s,t)$ is $\pi(s,t)$. Since $s\in \tau$, the edge sequence of $\pi_I(s,t)$ is in $\Lambda_{\tau}$. This proves that the edge sequence of $\pi(s,t)$ is in $\Lambda_{\tau}$. The fourth lemma statement thus follows. 
\end{proof}

\subsubsection{Proving the bound $\boldsymbol{k_{e'}=O(n^{3+\epsilon})}$}
\label{sec:bound}
We are now ready to prove the bound $k_{e'}=O(n^{3+\epsilon})$ using Lemma~\ref{lem:summaryf1}. We can obtain a lemma similar to Lemma~\ref{lem:summaryf1} with respect to the face $f_2$ (i.e., consider $f_2$-constrained paths). For each elementary segment $\tau$ of $e$ with respect to $f_2$, we define $\Lambda'_{\tau}$ in the same way as $\Lambda_{\tau}$ with respect to $f_1$. 

The endpoints of the elementary segments of $e$ with respect to $f_1$ and the endpoints of the elementary segments with respect to $f_2$ together partition $e$ into $O(n^2)$ segments (for ease of exposition, we assume that all these segment endpoints are distinct), which we call {\em refined-elementary segments}.  For each refined-elementary segment $\tau$, it is contained in an elementary segment $\tau_1$ for $f_1$ and an elementary segment $\tau_2$ for $f_2$; we define $\Lambda''_{\tau}=\Lambda_{\tau_1}\cup \Lambda'_{\tau_2}$. Since both $\Lambda_{\tau_1}$ and $\Lambda'_{\tau_2}$ are $O(n)$, we have $|\Lambda''_{\tau}|=O(n)$. We then have the following lemma. 

\begin{lemma}\label{lem:summary}
The edge $e$ can be partitioned into $O(n^2)$ refined-elementary segments such that the following properties hold.
\begin{enumerate}
    \item For each refined-elementary segment $\tau$, $|\Lambda_{\tau}''|=O(n)$. 
    \item  For every two adjacent refined-elementary segments $\tau$ and $\tau'$, $\Lambda''_{\tau}$ and $\Lambda''_{\tau'}$ differ by at most one edge sequence. 
    \item  For any refined-elementary segment $\tau$, for any point $s\in \tau$ and any edge sequence $\sigma\in \Lambda_{\tau}''$, there are a point $t\in e'$ and a geodesic \st\ path whose edge sequence is $\sigma$. 
    \item For any refined-elementary segment $\tau$, for any point $s\in \tau$ and any point $t\in e'$, the edge sequence of any shortest \st\ path is in $\Lambda''_{\tau}$. 
\end{enumerate}
\end{lemma}
\begin{proof}
We already discussed that $e$ can be partitioned into $O(n^2)$ refined-elementary segments and each such interval $\tau$ is completely contained in an elementary segment $\tau_1$ of $f_1$ and an elementary segment $\tau_2$ of $f_2$. Therefore, the first three lemma statements directly follow from Lemma~\ref{lem:summaryf1} (and its counterpart lemma with respect to $f_2$). For the fourth statement, consider a point $s\in \tau$, a point $t\in e'$, and a shortest path $\pi(s,t)$. The path $\pi(s,t)$ must be either $f_1$-constrained or $f_2$-constrained. Without loss of generality, we assume that it is $f_1$-constrained. Then, by Lemma~\ref{lem:summaryf1}, the edge sequence of $\pi(s,t)$ must be in $\Lambda_{\tau_1}$ and thus in $\Lambda''_{\tau}$ since $\Lambda_{\tau_1}\subseteq \Lambda''_{\tau}$. This proves the fourth lemma statement. 
\end{proof}

Let $A$ be the set of all refined-elementary segments of $e$. Let $\Pi_{e'}$ denote the plane containing the face $f_1$. For each segment $\tau\in A$ and each edge sequence $\sigma\in \Lambda_{\tau}''$, we unfold the edge sequence $\sigma$ into the plane $\Pi_{e'}$. For any point $s\in \tau$, let $s_{\sigma}$ be the image of $s$ in $\Pi_{e'}$ in the unfolding of $\sigma$. Define $S_{\tau}$ as the set of the images $s_{\sigma}$ for all edge sequences $\sigma\in \Lambda''_{\tau}$. 
In the proof of Lemma~\ref{lem:shortcut}, we show that there is no ``short-cut'' in $\Pi_{e'}$, i.e., for any point $s\in \tau$ and any point $t\in e'$, $\Vert s_{\sigma}t\Vert\geq d(s,t)$ holds for all $s_{\sigma}\in S_{\tau}$. 

\begin{lemma}\label{lem:shortcut}
For any point $s\in \tau$ and any point $t\in e'$, 
$\Vert s_{\sigma'}t\Vert =d(s,t)$ holds, where $\sigma'=\argmin_{\sigma\in \Lambda''_{\tau}}\Vert s_{\sigma}t\Vert$, i.e., $s_{\sigma'}$ is the point of $S_{\tau}$ closest to $t$.
\end{lemma}
\begin{proof}
On the one hand, by Lemma~\ref{lem:summary}(4), there is an edge sequence $\sigma\in \Lambda_{\tau}''$ such that $\sigma$ is the edge sequence of a shortest \st\ path $\pi(s,t)$. By the definition of $s_{\sigma}$, $\Vert s_{\sigma}t\Vert$ is equal to the length of $\pi(s,t)$. Thus, we have $\Vert s_{\sigma}t\Vert=d(s,t)$. Since $\Vert s_{\sigma'}t\Vert\leq \Vert s_{\sigma}t\Vert$, we obtain $\Vert s_{\sigma'}t\Vert\leq d(s,t)$. 

On the other hand, for any $\sigma\in \Lambda_{\tau}''$, by Lemma~\ref{lem:summary}(3), there are a point $t'\in e'$ and a geodesic $s$-$t'$ path $\pi$ whose edge sequence is $\sigma$. Hence, the image of $\pi$ in the unfolding of $\sigma$ is the line segment $\overline{s_{\sigma}t'}$ in $\Pi_{e'}$. Also, since $\overline{t't}$ is a line segment on $e'$, $\overline{t't}$ must be the shortest path between $t'$ and $t$ on $\calP$. According to \cite[Theorem 6]{ref:ChandruSh04}, $\Vert s_{\sigma}t\Vert\geq d(s,t)$ must hold. Since $\sigma'\in \Lambda_{\tau}''$, we obtain $\Vert s_{\sigma'}t\Vert\geq d(s,t)$.

We conclude that $\Vert s_{\sigma'}t\Vert= d(s,t)$. The lemma thus follows. 
\end{proof}


Back to our original problem of proving the bound $k_{e'}=O(n^{3+\epsilon})$, recall that an $e'$-event corresponds to a pair $(s,t)$ with $s\in e$ and $t\in e'$ such that there are three shortest \st\ paths. Let $\tau$ be the refined-elementary segment of $e$ that contains $s$. By Lemma~\ref{lem:shortcut}, the three shortest \st\ paths correspond to three closest points of $S_{\tau}$ to $t$. To determine the number of $e'$-event points in $\tau$, for each point $s_{\sigma}\in S_{\tau}$, we define $d_{\sigma}(s,t)=\Vert s_{\sigma}t\Vert$ with respect to $s\in \tau$ and $t\in e'$. As discussed in \cite{ref:AgarwalSt97}, the position of $s_{\sigma}$ in the plane $\Pi_{e'}$ is a linear function of $s\in \tau\subseteq e$. Hence, $d_{\sigma}(s,t)$ is a constant degree bivariate algebraic function of $s\in \tau$ and $t\in e'$. Therefore, the number of $e'$-event points in $\tau$ is no more than the number of vertices in the lower envelope of the functions $d_{\sigma}(s,t)$ for all $\sigma\in \Lambda_{\tau}''$. As $|\Lambda_{\tau}''|=O(n)$, the number of vertices in the lower envelope is bounded by $O(n^{2+\epsilon})$~\cite{ref:HalperinNe94,ref:SharirAl94}. As $e$ has $O(n^2)$ refined-elementary segments $\tau$, the number of $e'$-event points on $e$ is bounded by $O(n^{4+\epsilon})$. We can further improve this bound by a linear factor using a grouping technique similar to that in Section~\ref{sec:group} or in Section~\ref{sec:boundk1}, as follows. 

Let $F(\tau)$ denote the set of functions $d_{\sigma}(s,t)$ defined on $s\in \tau$ for all $\sigma\in \Lambda_{\tau}''$. Let $\tau'$ be a refined-elementary segment on $e$ adjacent to $\tau$. By Lemma~\ref{lem:summary}(2), $\Lambda_{\tau}''$ and $\Lambda_{\tau'}''$ differ by at most one edge sequence $\sigma$. Hence, $F(\tau)$ and $F(\tau')$ defer by at most one function (i.e., functions of $F(\tau)$ that are also in $F(\tau')$ can be extended from $s\in \tau$ to $s\in \tau\cup \tau'$ as $\tau\cup \tau'$ is a segment of $e'$). If we group every $n$ adjacent refined-elementary segments on $e$, then the number of functions in each group is $O(n)$ and the number of groups is $O(n)$. Since the number of functions in each group is $O(n)$, the number of vertices in the lower envelope of these functions is $O(n^{2+\epsilon})$~\cite{ref:HalperinNe94,ref:SharirAl94}, and thus 
the number of $e'$-event points in each group is $O(n^{2+\epsilon})$. Hence, the total number of $e'$-event points on $e$ is $O(n^{3+\epsilon})$. 
This proves the bound $k_{e'}=O(n^{3+\epsilon})$. 

\section{Two-point shortest path queries -- the edge-edge case}
\label{sec:edgeboth}

Our results in Section~\ref{sec:seq} may find other applications related to shortest path edge sequences on $\calP$. In this section, we demonstrate one application on the edge-edge case two-point shortest path query problem. In this case, for each pair of query points $(s,t)$, both $s$ and $t$ are required on edges of $\calP$. We show that we can obtain the same results as Theorems~\ref{theo:edgeface} and \ref{theo:singleedgeface} (with an additional constraint that $t$ is also on a polyhedron edge) but with a simpler approach. In the rest of this section, unless otherwise stated, both $s$ and $t$ are on edges of $\calP$. 


Consider a polyhedron edge $e$. We will build a data structure for queries $(s,t)$ with $s\in e$. Our overall data structure is comprised of all data structures for all polyhedron edges. Before describing our data structure for $e$, we first consider the one-point query problem in which $s$ is a fixed point on $e$ and only $t$ is a query point. We show that each query can be easily answered in $O(\log n)$ time by using its ridge tree $T_s$. Note that this is possible even if $t$ is not on a polyhedron edge~\cite{ref:SharirOn86,ref:MountOn85,ref:MitchellTh87}, but here we provide a simpler approach for the case where $t$ is on a polyhedron edge, and most importantly, later we will extend the solution to the two-point queries with the help of our results in Section~\ref{sec:seq}. 


\subsection{The one-point query problem}
Let $s$ be a fixed point on $e$. 
For each polyhedron edge $e'$, the ridge tree $T_s$ partitions it into a set $\calI(e')$ intervals such that for each interval $I\in \calI(e')$, the edge sequence of $\pi(s,t)$ for any point $t$ in the interior of $I$ is the same; let $\sigma_I$ denote the edge sequence. In the preprocessing, for each interval $I\in \calI(e')$, we compute the unfolding of $\sigma_I$ to a plane containing an incident face of $e'$ and store the image $s_I$ of $s$ in the unfolding. Doing this for all polyhedron edges can be done in $O(n^2\log n)$ time~\cite{ref:MountOn85,ref:MitchellTh87}. Since $T_s$ partitions all polyhedron edges into $O(n^2)$ intervals~\cite{ref:SharirOn86}, the total space is $O(n^2)$. Given a query point $t$ on a polyhedron edge $e'$, we first find the interval $I\in \calI(e')$ that contains $t$ by binary search and then return $d(s,t)=|s_It|$. Hence, computing $d(s,t)$ takes $O(\log n)$ time. 

To output a shortest path $\pi(s,t)$, we need to store the edge sequences $\sigma_I$ for all intervals $I\in \calI(e')$. To this end, by Lemma~\ref{lem:seqtree}, $\sigma_I$ is $\sigma_{v_I}$ for a node $v_I$ in the sequence tree $\Psi_s$ defined in Section~\ref{sec:generate}. Therefore, if we associate $I$ with the node $v_I$ in $\Psi_s$, then $\pi(s,t)$ can be output in $O(|\pi(s,t)|)$ time. To find $v_I$ for all intervals $I\in \calI(e')$ of all polyhedron edges $e'$, we do the following in the preprocessing. Consider a degree-3 vertex $u$ of $T_s$ and a shortest path $\pi(s,u)$. For each intersection $p$ between $\pi(s,u)$ and a polyhedron edge $e'$, since $e'$ is an edge in the edge sequence of $\pi(s,u)$, it corresponds to a node $v$ in $\Psi_s$, i.e., $\sigma_v$ is a prefix of the edge sequence of $\pi(s,u)$. 
We do binary search on $e'$ to find the interval $I\in \calI(e')$ containing $p$ and then set $v_I=v$. If we do this for shortest paths $\pi(s,u)$ of all degree-3 vertices $u$ of $T_s$, we claim that $v_I$ for all intervals $I\in \calI(e')$ of all polyhedron edges $e'$ are computed. Indeed, if $v_I$ for an interval $I\in \calI(e')$ is not computed, then the edge sequence of $\pi(s,t)$ for any interior point $t\in I$ is not a prefix of the edge sequence of $\pi(s,u)$ for any degree-3 vertex $u\in T_s$, a contradiction to Lemma~\ref{lem:seqtree}. Since $T_s$ has $O(n)$ degree-3 vertices $u$ and $\pi(s,u)$ intersects $O(n)$ polyhedron edges, the above algorithm takes $O(n^2\log n)$ time. Also, since the size of $\Psi_s$ is $O(n^2)$, the total space is $O(n^2)$. 

In summary, given a point $s\in e$, we can construct a data structure of $O(n^2)$ space in $O(n^2\log n)$ time such that for any query point $t$ on a polyhedron edge, $d(s,t)$ can be computed in $O(\log n)$ time and $\pi(s,t)$ can be output in additional $O(|\pi(s,t)|)$ time. 

\subsection{The two-point query problem}
We now consider the two-point queries $(s,t)$ with $s\in e$. If we move $s$ on $e$ slightly without changing $T_s$ topologically, then for each interval $I\in \calI(e')$ of any polyhedron edge $e'$, the location of either endpoint of $I$ on $e'$ is a constant-degree algebraic function of the position of $s\in e$. Suppose that $\tau$ is a segment on $e$ such that $T_s$ is topologically equivalent for all points $s$ in the interior of $\tau$; we call $\tau$ a {\em ridge-tree topologically-equivalent segment} (or {\em RTTE segment} for short). To answer queries $(s,t)$ with $s\in \tau$, we can parameterize the endpoints of the intervals of $I\in \calI(e')$ of all polyhedron edges $e'$ for $s\in \tau$. Since $\tau$ is an RTTE segment, the sequence tree $\Psi_s$ is the same for all points $s\in \tau$, and therefore we use $\Psi(\tau)$ to represent it. Furthermore, for each interval $I\in \calI(e')$ of each polyhedron edge $e'$, we store the coordinate transformation from $s\in e$ to the unfolding plane that contains $e'$ in the unfolding of the edge sequence $\sigma_I$ defined above, so that given a point $s\in e$, its image $s_I$ in the unfolding plane can be obtained in $O(1)$ time. The coordinate transformation is actually implicitly maintained in the sequence tree $\Psi(\tau)$. More specifically, if we have the coordinate transformation for a node $u\in \Psi_s$ (i.e., the transformation is for the edge sequence $\sigma_u$) and $v$ is a child of $u$, then the coordinate transformation for $v$ can be obtained from $u$ in $O(1)$ time. As such, we can use $\Psi(\tau)$ to explicitly store the coordinate transformations of all its nodes. For each interval $I\in \calI(e')$ of each polyhedron edge $e'$, its associated node $v_I\in \Psi(\tau)$ stores the coordinate transformation for $I$. 

Given a query $(s,t)$ with $s\in \tau$ and $t\in e'$ for a polyhedron edge $e'$, we can first determine the interval $I\in \calI(e')$ containing $t$ with respect to $s$. This can be done in $O(\log n)$ time by conducting binary search on the parameterized endpoints of the intervals of $\calI(e')$ on $e'$. After having $I$, using the coordinate transformation of $I$, which can be obtained from its associated node $v_I\in \Psi(\tau)$, we obtain the image $s_I$ of $s$ in the unfolding of $\sigma_I$ and consequently obtain $d(s,t)=\Vert s_It\Vert$. As before, computing $\pi(s,t)$ can be done in additional $O(|\pi(s,t)|)$ time using the node $v_I$ and $\Psi(e)$. 

By using the algorithm in Section~\ref{sec:maintain}, we can partition $e$ into $O(n^{4+\epsilon})$ RTTE segments in $O(n^{4+\epsilon})$ time. Indeed, the algorithm computes $O(n^{4+\epsilon})$ event points on $e$ such that when $s$ is between any two adjacent event points, $T_s$ does not change topologically. Therefore, these event points partition $e$ into  $O(n^{4+\epsilon})$ RTTE segments. Computing these event points takes $O(n^{4+\epsilon})$ time.

If we build the above data structure for each RTTE segment $\tau$, then the total space would be $O(n^{6+\epsilon})$. To improve it, we resort to persistent data structures~\cite{ref:DriscollMa89}. Specifically, we build the data structure for the first RTTE segment $\tau$ of $e$. For each polyhedron edge $e'$, we construct a partially persistent tree $T(e')$ on the parameterized endpoints of the intervals of $\calI(e')$ so that after updates of $T(e')$ we can still perform queries in old versions of it. We use $\calI_{\tau}(e')$ to refer to $\calI(e')$, meaning that it is defined with respect to $\tau$. Now consider the second interval $\tau'$. Let $p$ be the common endpoint of $\tau$ and $\tau'$. As discussed in Section~\ref{sec:seq}, $p$ is either the first or the second type event point. 

\begin{enumerate}
\item If $p$ is the first type event point, then for any polyhedron edges $e'$, $\calI_{\tau'}(e')$ is {\em topologically equivalent} to $\calI_{\tau}(e')$, i.e., each interval $I\in \calI_{\tau'}(e')$ corresponds to an interval $I'\in \calI_{\tau}(e')$ such that their edge sequences $\sigma_I$ and $\sigma_{I'}$ are the same. Therefore, in this case we do not need to update our data structure, i.e., we use the same $T(e')$ for each polyhedron edge $e'$. 

\item 
If $p$ is the second type event point, then as discussed in Section~\ref{sec:seq}, $\calI_{\tau}(e')$ is not topologically equivalent to $\calI_{\tau'}(e')$ for only one polyhedron edge $e'$, and $\calI_{\tau}(e')$ and $\calI_{\tau'}(e')$ differ by at most one interval (i.e., either one interval of $\calI_{\tau}(e')$ disappears in $\calI_{\tau'}(e')$, or a new interval is inserted to $\calI_{\tau'}(e')$). If a new interval $I$ is inserted to $\calI_{\tau'}(e')$, recall that by using the sequence tree $\Psi(e)$ our algorithm in Section~\ref{sec:maintain} also obtains in $O(1)$ time the unfolding of the edge sequence $\sigma_I$ and thus the coordinate transformation from $s\in e$ to the unfolding plane that contains $e'$ in the unfolding of $\sigma_I$ (more specifically, as discussed above, if a node $v$ representing the edge sequence $\sigma_I$ of $I$ is added to $\Psi(e)$, then we can obtain the coordinate transformation for $I$ in $O(1)$ by using the coordinate transformation of its parent node in $\Psi(e)$); we store the coordinate transformation for $I$ (so that given $s\in \tau'$ and any point $t\in I$, we can compute the image $s_I$ of $s$ in the unfolding in $O(1)$ time and consequently obtain $d(s,t)=\Vert s_It\Vert$). As such, we can update $T(e')$ in $O(\log n)$ time by keeping its old version (for $\tau$) so that a query on the old version of the tree can still be answered in $O(\log n)$ time~\cite{ref:DriscollMa89}. We do this for all other RTTE segments on $e$. Therefore, constructing the data structures for all $O(n^{4+\epsilon})$ RTTE segments can be done using $O(n^{4+\epsilon})$ insertions and deletions on persistent trees $T(e')$, which takes $O(n^{4+\epsilon})$ time and space in total~\cite{ref:DriscollMa89}. 
\end{enumerate}

Given a query $(s,t)$ with $s\in e$ and $t\in e'$ for a polyhedron edge $e'$, we first find the RTTE segment $\tau$ that contains $s$ by binary search. Then, we search the tree $T(e')$ of the version corresponding to $\tau$ to find the interval $I\in \calI_{\tau}(e')$ that contains $t$. Finally, using the coordinate transformation stored at $I$, we obtain the image $s_I$ of $s$ in the unfolding of $\sigma_I$ and then return $d(s,t)=\Vert s_It\Vert$. As such, it takes $O(\log n)$ time to compute $d(s,t)$. 

The above only computes $d(s,t)$. To report the shortest path $\pi(s,t)$, as before, we will use the sequence tree $\Psi_s$. To this end, we first compute $\Psi_s$ when $s$ is at an endpoint of $e$ and then update $\Psi_s$ as $s$ moves but keep the old versions of $\Psi_s$. This is exactly the tree $\Psi(e)$ we computed in Section~\ref{sec:maintain}. In addition, we maintain the coordinate transformations for all nodes of $\Psi(e)$ as discussed above, which adds $O(1)$ additional time whenever a node is updated. 
Furthermore, we also need to associate $I$ to its node $v_I$ in $\Psi(e)$ once a new interval $I$ is created on an edge $e'$, which can only happen at a second type event point. This node $v_I$ is exactly the new node associated with $v$, where $v$ is the degree-3 vertex that generates the second type event point (i.e., when $v$ crosses $e'$; e.g. in Figure~\ref{fig:secondtypeevent}(3), a new interval $I=\overline{u_1u_2}$ is created and $I_v$ is $c'$, the node associated with $v$ right after the event discussed in Section~\ref{sec:mainloop}). As discussed in Section~\ref{sec:mainloop}, each update on $\Psi(e)$ takes $O(1)$ time and space, and our additional steps also take $O(1)$ time and space. Therefore, the preprocessing still takes $O(n^{4+\epsilon})$ time and space. Now for each query $(s,t)$ with $s\in e$ and $t\in e'$, after the RTTE segment $\tau$ of $e$ containing $s$ and the interval $I\in \calI_{\tau}(e')$ containing $t$ are computed, using the node $v_I$ in $\Psi(e)$ and following the path from $v_I$ to the root of $\Psi(e)$, we can compute $\pi(s,t)$ in additional $O(|\pi(s,t)|)$ time. 

The following theorem summarizes our main result in this section. 

\begin{theorem}
Given a convex polyhedral surface $\calP$ of $n$ vertices, for each edge $e$ of $\calP$, we can construct a data structure for $e$ in $O(n^{4+\epsilon})$ preprocessing time and space such that for any two query points $s\in e$ and $t$ on an edge of $\calP$, the geodesic distance $d(s,t)$ can be computed in $O(\log n)$ time and an actual shortest path $\pi(s,t)$ can be output in additional $O(|\pi(s,t)|)$ time. 
\end{theorem}

Applying the above theorem to each edge $e$ of $\calP$ leads to the following. 

\begin{corollary}
Given a convex polyhedral surface $\calP$ of $n$ vertices, we can construct a data structure in $O(n^{5+\epsilon})$ preprocessing time and space such that for any two query points $s$ and $t$ both on edges of $\calP$, the geodesic distance $d(s,t)$ can be computed in $O(\log n)$ time and a shortest path $\pi(s,t)$ can be reported in additional $O(|\pi(s,t)|)$ time.    
\end{corollary}


%
\bibliographystyle{plainurl}
\bibliography{reference}

\begin{thebibliography}{10}

\bibitem{ref:AgarwalSt97}
Pankaj~K. Agarwal, Boris Aronov, Joseph O'Rourke, and Catherine~A. Schevon.
\newblock Star unfolding of a polytope with applications.
\newblock {\em SIAM Journal on Computing}, 26:1689--1713, 1997.
\newblock \href {https://doi.org/10.1137/S0097539793253371} {\path{doi:10.1137/S0097539793253371}}.

\bibitem{ref:AgarwalCo97}
Pankaj~K. Agarwal, Boris Aronov, and Micha Sharir.
\newblock Computing envelopes in four dimensions with applications.
\newblock {\em SIAM Journal on Computing}, 26:1714--1732, 1997.
\newblock \href {https://doi.org/10.1137/S0097539794265724} {\path{doi:10.1137/S0097539794265724}}.

\bibitem{ref:AgarwalRa93}
Pankaj~K. Agarwal and Ji\u{r}\'{i} Matou\v{s}ek.
\newblock Ray shooting and parametric search.
\newblock {\em SIAM Journal on Computing}, 22(4):794--806, 1993.
\newblock \href {https://doi.org/10.1137/0222051} {\path{doi:10.1137/0222051}}.

\bibitem{ref:AgarwalPs05}
Pankaj~K. Agarwal and Micha Sharir.
\newblock Pseudoline arrangements: Duality, algorithms, and applications.
\newblock {\em SIAM Journal on Computing}, 34:526--552, 2005.
\newblock \href {https://doi.org/10.1137/S0097539703433900} {\path{doi:10.1137/S0097539703433900}}.

\bibitem{ref:AronovNo92}
Boris Aronov and Joseph O'Rourke.
\newblock Nonoverlap of the star unfolding.
\newblock {\em Discrete and Computational Geometry}, 8:219--250, 1992.
\newblock \href {https://doi.org/10.1007/BF02293047} {\path{doi:10.1007/BF02293047}}.

\bibitem{ref:BenderTh00}
Michael~A. Bender and Mart\'{i}n Farach-Colton.
\newblock The {LCA} problem revisited.
\newblock In {\em Proceedings of the 4th Latin American Symposium on Theoretical Informatics (LATIN)}, pages 88--94, 2000.
\newblock \href {https://doi.org/10.1007/10719839_9} {\path{doi:10.1007/10719839_9}}.

\bibitem{ref:deBergTo24}
Sarita~de Berg, Tillmann Miltzow, and Frank Staals.
\newblock Towards space efficient two-point shortest path queries in a polygonal domain.
\newblock In {\em Proceedings of the 40th International Symposium on Computational Geometry (SoCG)}, pages 17:1--16, 2024.
\newblock \href {https://doi.org/10.4230/LIPIcs.SoCG.2024.17} {\path{doi:10.4230/LIPIcs.SoCG.2024.17}}.

\bibitem{ref:CannyNe87}
John Canny and John Reif.
\newblock New lower bound techniques for robot motion planning problems.
\newblock In {\em Proceedings of the 28th Annual Symposium on Foundations of Computer Science (FOCS)}, pages 49--60, 1987.
\newblock \href {https://doi.org/10.1109/SFCS.1987.42} {\path{doi:10.1109/SFCS.1987.42}}.

\bibitem{ref:ChandruSh04}
Vijay Chandru, Ramesh Hariharan, and Narasimha~M. Krishnakumar.
\newblock Short-cuts on star, source and planar unfoldings.
\newblock In {\em Proceedings of the 24th International Conference on Foundations of Software Technology and Theoretical Computer Science (FSTTCS)}, pages 174--185, 2004.
\newblock \href {https://doi.org/10.1007/978-3-540-30538-5_15} {\path{doi:10.1007/978-3-540-30538-5_15}}.

\bibitem{ref:ChazelleCu93}
Bernard Chazelle.
\newblock Cutting hyperplanes for divide-and-conquer.
\newblock {\em Discrete and Computational Geometry}, 9(2):145--158, 1993.
\newblock \href {https://doi.org/10.1007/BF02189314} {\path{doi:10.1007/BF02189314}}.

\bibitem{ref:ChenSh00}
Danny~Z. Chen, K.S. Klenk, and H.-Y.T. Tu.
\newblock Shortest path queries among weighted obstacles in the rectilinear plane.
\newblock {\em SIAM Journal on Computing}, 29(4):1223--1246, 2000.

\bibitem{ref:ChenSh90}
Jindong Chen and Yijie Han.
\newblock Shortest paths on a polyhedron.
\newblock In {\em Proceedings of the sixth Annual Symposium on Computational geometry (SoCG)}, pages 360--369, 1990.
\newblock \href {https://doi.org/10.1145/98524.98601} {\path{doi:10.1145/98524.98601}}.

\bibitem{ref:ChengSh14}
Siu-Wing Cheng and Jiongxin Jin.
\newblock Shortest paths on polyhedral surfaces and terrains.
\newblock In {\em Proceedings of the 46th Annual ACM Symposium on Theory of Computing (STOC)}, pages 373--382, 2014.
\newblock \href {https://doi.org/10.1145/2591796.2591821} {\path{doi:10.1145/2591796.2591821}}.

\bibitem{ref:ChiangTw99}
Yi-Jen Chiang and Joseph S.~B. Mitchell.
\newblock Two-point {Euclidean} shortest path queries in the plane.
\newblock In {\em Proceedings of the Annual ACM-SIAM Symposium on Discrete Algorithms (SODA)}, pages 215--224, 1999.
\newblock \href {https://doi.org/10.5555/314500.314560} {\path{doi:10.5555/314500.314560}}.

\bibitem{ref:DriscollMa89}
James~R. Driscoll, Neil Sarnak, Daniel~D. Sleator, and Robert~E. Tarjan.
\newblock Making data structures persistent.
\newblock {\em Journal of Computer and System Sciences}, 38(1):86--124, 1989.
\newblock \href {https://doi.org/10.1016/0022-0000(89)90034-2} {\path{doi:10.1016/0022-0000(89)90034-2}}.

\bibitem{ref:EdelsbrunnerOp86}
Herbert Edelsbrunner, Leonidas~J. Guibas, and Jorge Stolfi.
\newblock Optimal point location in a monotone subdivision.
\newblock {\em SIAM Journal on Computing}, 15(2):317--340, 1986.
\newblock \href {https://doi.org/10.1137/0215023} {\path{doi:10.1137/0215023}}.

\bibitem{ref:GhoshAn91}
Subir~K. Ghosh and David~M. Mount.
\newblock An output-sensitive algorithm for computing visibility graphs.
\newblock {\em SIAM Journal on Computing}, 20(5):888--910, 1991.
\newblock \href {https://doi.org/10.1137/0220055} {\path{doi:10.1137/0220055}}.

\bibitem{ref:GuibasLi87}
Leonidas~J. Guibas, John Hershberger, D.~Leven, Micha Sharir, and Robert~E. Tarjan.
\newblock Linear-time algorithms for visibility and shortest path problems inside triangulated simple polygons.
\newblock {\em Algorithmica}, 2(1-4):209--233, 1987.
\newblock \href {https://doi.org/10.1007/BF01840360} {\path{doi:10.1007/BF01840360}}.

\bibitem{ref:GuibasVo91}
Leonidas~J. Guibas, Joseph S.~B. Mitchell, and Thomas Roos.
\newblock Voronoi diagrams of moving points in the plane.
\newblock In {\em Proceedings of the 17th International Workshop on Graph-Theoretic Concepts in Computer Science (WG)}, pages 113--125, 1991.
\newblock \href {https://doi.org/10.1007/3-540-55121-2_11} {\path{doi:10.1007/3-540-55121-2_11}}.

\bibitem{ref:HalperinNe94}
Dan Halperin and Micha Sharir.
\newblock New bounds for lower envelopes in three dimensions, with applications to visibility in terrains.
\newblock {\em Discrete and Computational Geometry}, 12:313--326, 1994.
\newblock \href {https://doi.org/10.1007/BF02574383} {\path{doi:10.1007/BF02574383}}.

\bibitem{ref:HarelFa84}
Dov Harel and Robert~E. Tarjan.
\newblock Fast algorithms for finding nearest common ancestors.
\newblock {\em SIAM Journal on Computing}, 13:338--355, 1984.
\newblock \href {https://doi.org/10.1137/0213024} {\path{doi:10.1137/0213024}}.

\bibitem{ref:HershbergerAn99}
John Hershberger and Subhash Suri.
\newblock An optimal algorithm for {Euclidean} shortest paths in the plane.
\newblock {\em SIAM Journal on Computing}, 28(6):2215--2256, 1999.
\newblock \href {https://doi.org/10.1137/S0097539795289604} {\path{doi:10.1137/S0097539795289604}}.

\bibitem{ref:HWangFi89}
Yie-Huei Hwang, Ruei-Chuan Chang, and Hung-Yi Tu.
\newblock Finding all shortest path edge sequences on a convex polyhedron.
\newblock In {\em Proceedings of the 1st Workshop on Algorithms and Data Structures (WADS)}, pages 251--266, 1989.
\newblock \href {https://doi.org/10.1007/3-540-51542-9_23} {\path{doi:10.1007/3-540-51542-9_23}}.

\bibitem{ref:CookSh12}
Atlas F.~Cook IV and Carola Wenk.
\newblock Shortest path problems on a polyhedral surface.
\newblock {\em Algorithmica}, 69:58--77, 2012.
\newblock \href {https://doi.org/10.1007/s00453-012-9723-6} {\path{doi:10.1007/s00453-012-9723-6}}.

\bibitem{ref:KapoorEf99}
Sanjiv Kapoor.
\newblock Efficient computation of geodesic shortest paths.
\newblock In {\em Proceedings of the 31st Annual ACM Symposium on Theory of Computing (STOC)}, pages 770--779, 1999.
\newblock \href {https://doi.org/10.1145/301250.301449} {\path{doi:10.1145/301250.301449}}.

\bibitem{ref:KapoorAn97}
Sanjiv Kapoor, S.N. Maheshwari, and Joseph S.~B. Mitchell.
\newblock An efficient algorithm for {Euclidean} shortest paths among polygonal obstacles in the plane.
\newblock {\em Discrete and Computational Geometry}, 18(4):377--383, 1997.
\newblock \href {https://doi.org/10.1007/PL00009323} {\path{doi:10.1007/PL00009323}}.

\bibitem{ref:KirkpatrickOp83}
David~G. Kirkpatrick.
\newblock Optimal search in planar subdivisions.
\newblock {\em SIAM Journal on Computing}, 12(1):28--35, 1983.
\newblock \href {https://doi.org/10.1137/0212002} {\path{doi:10.1137/0212002}}.

\bibitem{ref:MatousekRay93}
Ji\u{r}\'{i} Matou\v{s}ek and Otfried Schwarzkopf.
\newblock Ray shooting in convex polytopes.
\newblock {\em Discrete and Computational Geometry}, 10:215--232, 1993.
\newblock \href {https://doi.org/10.1007/BF02573975} {\path{doi:10.1007/BF02573975}}.

\bibitem{ref:MitchellA91}
Joseph S.~B. Mitchell.
\newblock A new algorithm for shortest paths among obstacles in the plane.
\newblock {\em Annals of Mathematics and Artificial Intelligence}, 3(1):83--105, 1991.
\newblock \href {https://doi.org/10.1007/BF01530888} {\path{doi:10.1007/BF01530888}}.

\bibitem{ref:MitchellSh96}
Joseph S.~B. Mitchell.
\newblock Shortest paths among obstacles in the plane.
\newblock {\em International Journal of Computational Geometry and Applications}, 6(3):309--332, 1996.
\newblock \href {https://doi.org/10.1142/S0218195996000216} {\path{doi:10.1142/S0218195996000216}}.

\bibitem{ref:MitchellTh87}
Joseph S.~B. Mitchell, David~M. Mount, and Christos~H. Papadimitriou.
\newblock The discrete geodesic problem.
\newblock {\em SIAM Journal on Computing}, 16:647--668, 1987.
\newblock \href {https://doi.org/10.1137/0216045} {\path{doi:10.1137/0216045}}.

\bibitem{ref:MountOn85}
David~M. Mount.
\newblock On finding shortest paths on convex polyhedra.
\newblock Technical report, University of Maryland, College Park, MD 20742, 1985.
\newblock \url{https://www.cs.umd.edu/~mount/Papers/tr85-sp-convex.pdf}.

\bibitem{ref:MountSt87}
David~M. Mount.
\newblock Storing the subdivision of a polyhedral surface.
\newblock {\em Discrete and Computational Geometry}, 2:153--174, 1987.
\newblock \href {https://doi.org/10.1007/BF02187877} {\path{doi:10.1007/BF02187877}}.

\bibitem{ref:MountTh90}
David~M. Mount.
\newblock The number of shortest paths on the surface of a polyhedron.
\newblock {\em SIAM Journal on Computing}, 19:593--611, 1990.
\newblock \href {https://doi.org/10.1137/0219040} {\path{doi:10.1137/0219040}}.

\bibitem{ref:ORourkeCo89}
Joseph O'Rourke and Catherine~A. Schevon.
\newblock Computing the geodesic diameter of a 3-polytope.
\newblock In {\em Proceedings of the 15th Annual Symposium on Computational Geometry (SoCG)}, pages 370--379, 1989.
\newblock \href {https://doi.org/10.1145/73833.73874} {\path{doi:10.1145/73833.73874}}.

\bibitem{ref:RohnertSh86}
Hans Rohnert.
\newblock Shortest paths in the plane with convex polygonal obstacles.
\newblock {\em Information Processing Letters}, 23(2):71--76, 1986.
\newblock \href {https://doi.org/10.1016/0020-0190(86)90045-1} {\path{doi:10.1016/0020-0190(86)90045-1}}.

\bibitem{ref:SarnakPl86}
Neil Sarnak and Robert~E. Tarjan.
\newblock Planar point location using persistent search trees.
\newblock {\em Communications of the ACM}, 29:669--679, 1986.
\newblock \href {https://doi.org/10.1145/6138.6151} {\path{doi:10.1145/6138.6151}}.

\bibitem{ref:SchevonAl89}
Catherine~A. Schevon.
\newblock {\em Algorithms for Geodesics on Polytopes}.
\newblock PhD thesis, {Johns Hopkins University, Baltimore, MD}, 1989.

\bibitem{ref:SchevonTh88}
Catherine~A. Schevon and Joseph O'Rourke.
\newblock An algorithm for finding edge sequences on a polytope.
\newblock Technical Report JHU-89/03, Department of Computer Science, Johns Hopkins University, Baltimore, MD, 1988.

\bibitem{ref:SchreiberAn08}
Yevgeny Schreiber and Micha Sharir.
\newblock An optimal-time algorithm for shortest paths on a convex polytope in three dimensions.
\newblock {\em Discrete and Computational Geometry}, 39:500--579, 2008.
\newblock \href {https://doi.org/10.1007/s00454-007-9031-0} {\path{doi:10.1007/s00454-007-9031-0}}.

\bibitem{ref:SharirOn87}
Micha Sharir.
\newblock On shortest paths amidst convex polyhedra.
\newblock {\em SIAM Journal on Computing}, 16:561--572, 1987.
\newblock \href {https://doi.org/10.1137/0216038} {\path{doi:10.1137/0216038}}.

\bibitem{ref:SharirAl94}
Micha Sharir.
\newblock Almost tight upper bounds for lower envelopes in higher dimensions.
\newblock {\em Discrete and Computational Geometry}, 12:327--345, 1994.
\newblock \href {https://doi.org/10.1007/BF02574384} {\path{doi:10.1007/BF02574384}}.

\bibitem{ref:SharirOn86}
Micha Sharir and Amir Schorr.
\newblock On shortest paths in polyhedral spaces.
\newblock {\em SIAM Journal on Computing}, 15:193--215, 1986.
\newblock \href {https://doi.org/10.1145/800057.808676} {\path{doi:10.1145/800057.808676}}.

\bibitem{ref:VaradarajanAp00}
Kasturi~R. Varadarajan and Pankaj~K. Agarwal.
\newblock Approximating shortest paths on a non-convex polyhedron.
\newblock {\em SIAM Journal on Computing}, 30:1321--1340, 2000.
\newblock \href {https://doi.org/10.1137/S0097539799352759} {\path{doi:10.1137/S0097539799352759}}.

\bibitem{ref:WangA23}
Haitao Wang.
\newblock A new algorithm for {Euclidean} shortest paths in the plane.
\newblock {\em Journal of the ACM}, 70:11:1--11:62, 2023.
\newblock \href {https://doi.org/10.1145/3580475} {\path{doi:10.1145/3580475}}.

\bibitem{ref:WangUn23}
Haitao Wang.
\newblock Unit-disk range searching and applications.
\newblock {\em Journal of Computational Geometry}, 14:343--394, 2023.
\newblock \href {https://doi.org/10.20382/jocg.v14i1a13} {\path{doi:10.20382/jocg.v14i1a13}}.

\bibitem{ref:WangSh25}
Haitao Wang.
\newblock Shortest paths on convex polyhedral surfaces.
\newblock In {\em Proceedings of 66th IEEE Symposium on Foundations of Computer Science (FOCS)}, 2025.

\end{thebibliography}
\end{document}